\documentclass[12pt]{article}
\usepackage{fullpage}
\usepackage{graphicx}
\usepackage{array}
\usepackage[round]{natbib}
\usepackage{amsmath}
\usepackage{amsfonts}
\usepackage{amsthm}
\usepackage{hyperref}
\hypersetup{colorlinks,breaklinks,
            urlcolor=[rgb]{0,0.5,0.5},
            linkcolor=[rgb]{0,0.5,0.5},
            citecolor=[rgb]{0,0.2,.8}}
\usepackage{algorithm}
\usepackage{algpseudocode}
\usepackage{algorithmicx}
\usepackage{color}
\usepackage{tikz}
\usetikzlibrary{arrows,calc}
\usepackage{subcaption}
\usepackage{ulem}
\usepackage{authblk}
\usepackage{enumerate}
\usepackage{dsfont}
\usepackage{epstopdf}

\def\eg{\textit{e.g.\,}}

\newcommand{\x}{{\mathbf{x}}}

\newcommand{\y}{{\mathbf{y}}}

\newcommand{\ex}{\mathbb{E}}

\newcommand{\tdot}{\dot{\theta}}

\def\Yset{\mathcal{Y}}
\def\rset{\mathbb{R}}

\def\balpha{\bar{\alpha}}
\def\iid{\textit{iid}\,}
\def\Sset{\mathcal{S}}

\def\ie{\text{i.e.}\,}

\def\sfrak{\mathfrak{s}}
\def\Sfrak{\mathfrak{S}}

\def\nset{\mathbb{N}}
\def\rmd{\mathrm{d}}
\def\bP{\bar{P}}
\def\T{^\text{T}}
\def\Gfrak{\mathfrak{G}}
\def\Ufrak{\mathfrak{U}}
\def\epsi{\epsilon}

\def\var{\mathrm{var}}
\def\cov{\mathrm{cov}}

\def\barp{\bar{p}}

\def\1{\mathds{1}}
\def\bigO{\mathcal{O}}
\def\FP{\mathrm{FP}}
\def\DP{\mathrm{DP}}
\def\OP{\mathrm{OP}}
\def\bR{\bar{R}}
\def\bX{\bar{X}}
\def\norm{\mathcal{N}}
\def\Pfrak{\mathfrak{p}}
\def\hpi{\hat{\pi}}
\def\proba{\mathbb{P}}

\newcommand{\eps}{\varepsilon}

\newtheorem{thm}{Theorem}
\newtheorem{prop}[thm]{Proposition}
\newtheorem{lemma}[thm]{Lemma}
\newtheorem{cor}{Corollary}

\newtheorem{rmk}{Remark}

\newtheorem{exa}{Example}

\title{Efficient MCMC for Gibbs Random Fields using pre-computation}
\author[1,2]{Aidan Boland}
\author[1,2]{Nial Friel}
\author[1,2]{Florian Maire}
\affil[1]{School of Mathematics and Statistics, University College Dublin}
\affil[2]{Insight Centre for Data Analytics}
\date{}
\begin{document}
\maketitle

\abstract

Bayesian inference of Gibbs random fields (GRFs) is often referred to as a doubly intractable problem, since the likelihood function is intractable. The exploration of the posterior distribution of
such models is typically carried out with a sophisticated Markov chain Monte Carlo (MCMC) method, the exchange algorithm \citep{murray06}, which requires simulations from the likelihood function
at each iteration. The purpose of this paper is to consider an approach to dramatically reduce this computational overhead.
To this end we introduce a novel class of algorithms which use realizations of the GRF model, simulated offline, at locations specified by a grid that spans the parameter space. This strategy speeds
up dramatically the posterior inference, as illustrated on several examples. However, using the pre-computed graphs introduces a noise in the MCMC algorithm, which is no longer exact. We study the
theoretical behaviour of the resulting approximate MCMC algorithm and derive convergence bounds using a recent theoretical development on approximate MCMC methods.

\section{Introduction}

The focus of this study is on Bayesian inference of Gibbs random fields (GRFs), a class of models used in many areas of statistics, such as the autologistic model \citep{besag74} in spatial statistics,
the exponential random graph model in social network analysis \citep{robins07}, etc. Unfortunately, for all but trivially small graphs, GRFs suffer from intractability of the likelihood function making
standard analysis impossible. Such models are often referred to as \textit{doubly-intractable} in the Bayesian literature, since the normalizing constant of both the likelihood function and the posterior
distribution form a source of intractability.
In the recent past there has been considerable research activity in designing Bayesian algorithms which overcome this intractability all of which rely on simulation from the intractable likelihood.
Such methods include Approximate Bayesian Computation initiated by \citet{pritchard99} (see \eg \citet{marin2012approximate} for an excellent review) and Pseudo-Marginal algorithms \citep{AndRob}. Perhaps the
most popular approach to infer a doubly-intractable posterior distribution is the exchange algorithm \citep{murray06}. The exchange algorithm is a Markov chain Monte Carlo (MCMC) method that extends the
Metropolis-Hastings (MH) algorithm \citep{metropolis1953equation} to situations where the likelihood is intractable.
Compared to MH, the exchange uses a different acceptance probability and this has two main implications:
\begin{itemize}
\item theoretically: the exchange chain is less efficient than the MH chain, in terms of mixing time and asymptotic variance (see \cite{peskun1973optimum} and \cite{tierney1998note} for a discussion on the optimality of the MH chain)
\item computationally: at each iteration, the exchange requires \textit{exact} and \textit{independent} draws from the likelihood model at the current state of the Markov chain to calculate the acceptance probability,
a step that may substantially impact upon the computational performance of the algorithm
\end{itemize}
For many likelihood models, it is not possible to simulate \textit{exactly} from the likelihood function. In those situations, \citet{cucala09} and \citet{caimo11} replace the exact sampling step in the exchange algorithm with the simulation of an auxiliary Markov chain targeting the likelihood function, whereby inducing a noise process in the main Markov chain. This approximation was extended further by \citet{alquier16} who used multiple samples to speed up the convergence of the exchange algorithm.

This short literature review of the exchange algorithm and its variants shows that simulations from the likelihood function, either exactly or approximately, is central to those methods. However, this simulation step often compromises their practical implementation, especially for large graph models. Indeed, for a realistic run time, a user may end up with a limited number of draws from the posterior as most of the computational budget is dedicated to obtaining likelihood realizations. In addition, note that since the likelihood draws are conditioned on the Markov chain states, those simulation steps are intrinsically incompatible with parallel computing \citep{friel2016exploiting}.

Intuitively, there is a redundance of simulation. Indeed, should the Markov chain return to an area previously visited, simulation of the likelihood is nevertheless carried out as it had never been done before. This
is precisely the point we address in this paper.  We propose a novel class of algorithms where likelihood realizations are generated and then subsequently re-used at in an online inference phase.
More precisely, a regular grid spanning the parameter space is specified and draws from the likelihood at locations given by the vertices of this grid are obtained offline in a parallel fashion. The grid is tailored
to the posterior topology using estimators of the gradient and the Hessian matrix to ensure that the pre-computation sampling covers the posterior areas of high probability. However, using realizations of the
likelihood at pre-specified grid points instead of at the actual Markov chain state introduces a noise process in the algorithm. This leads us to study the theoretical behaviour of the resulting approximate MCMC algorithm
and to derive quantitative convergence bounds using the noisy MCMC framework developed in \citet{alquier16}. Essentially, our results allow one to quantify how the noise induced by the pre-computing step propagates through to the stationary distribution of the approximate chain. We find an upper bound on the bias between this distribution and the posterior of interest, which depends on the pre-computing step parameters \ie the distance between the grid points and the number of graphs drawn at each grid point. We also show that the bias vanishes asymptotically in the number of simulated graphs at each grid point, regardless of the grid structure.

Note that \citet{moores14} suggested a similar strategy to speed-up ABC algorithms by learning about the sufficient statistics of simulated data through an
estimated mapping function that uses draws from the likelihood function at a pre-defined set of parameter values. This method was shown to be computationally
very efficient but its suitability for models with more than one parameter can be questioned. Finally, we note that a related approach has been presented by
\citet{everitt17} which also relies on previously sampled likelihood draws in order to estimate the intractable ratio of normalising constants. However this
approach falls within a sequential Monte Carlo framework.

The paper is organised as follows. Section \ref{sec:methods} introduces the intractable likelihood that we focus on and details our class of approximate MCMC schemes which uses pre-computed likelihood simulations.
We also detail how we automatically specific the grid of parameter values.
In Section \ref{sec:theory}, we establish some theoretical results for noisy MCMC algorithms making use of a pre-computation step. In Section \ref{sec:results},
the inference of a number of GRFs is carried out using both pre-computed algorithms and exact algorithms such as the exchange. Results show a dramatic improvement of our method over exact methods in time normalized
experiments. Finally, this paper concludes with some related open problems.

\section{Pre-computing Metropolis algorithms}\label{sec:methods}

\subsection{Preliminary notation}
We frame our analysis in the setting of Gibbs random fields (GRFs) and we denote by $y\in \Yset$ the observed graph. A graph is identified by its adjacency matrix and $\Yset$ is taken as $\Yset:=\{0,1\}^{p\times p}$ where $p$ is the number of nodes in the graph. The likelihood function of $y$ is paramaterized by a vector $\theta\in\Theta\subset\rset^d$ and is defined as
$$
f(y|\theta)=\dfrac{q_{\theta}(y)}{Z(\theta)}=\dfrac{\exp\{\theta\T s(y)\}}{Z(\theta)},
$$
where $s(y)\in\Sset\subset\rset_+^{d}$ is a vector of statistics which are sufficient for the likelihood.
The normalizing constant,
\begin{align*}
Z(\theta)=\sum_{y\in\Yset}\exp\{\theta\T s(y)\},
\end{align*}
depends on $\theta$ and is intractable for all but trivially small graphs. The aim is to infer the parameters $\theta$ through the posterior distribution
\begin{align*}
\pi(\theta\,|\,y)\propto \frac{q_{\theta}(y)}{Z(\theta)}p(\theta),
\end{align*}
where $p$ denotes the prior distribution of $\theta$. In absence of ambiguity, a distribution and its probability density function will share the same notation.

\subsection{Computational complexity of MCMC algorithms for doubly intractable distributions}

In Bayesian statistics, Markov chain Monte Carlo methods (MCMC, see \eg \cite{gilks1995markov} for an introduction) remain the most popular way to explore $\pi$. MCMC algorithms proceed by creating a Markov chain whose invariant distribution has a density equal to the posterior distribution. One such algorithm, the Metropolis-Hastings (MH) algorithm \cite{metropolis1953equation}, creates a Markov chain by sequentially drawing candidate parameters from a proposal distribution $\theta'\sim h(\,\cdot\,|\theta)$ and accepting the proposed new parameter $\theta'$ with probability
\begin{align}
\alpha(\theta,\theta'):=1\wedge a(\theta,\theta')\,,\qquad a(\theta,\theta'):=\dfrac{q_{\theta'}(y)p(\theta')h(\theta|\theta')}{q_{\theta}(y)p(\theta)h(\theta'|\theta)}\times\dfrac{Z(\theta)}{Z(\theta')}\,.
\label{eqn:MHratio}
\end{align}
This acceptance probability depends on the ratio $Z(\theta)/Z(\theta')$ of the intractable normalising constants and cannot therefore be calculated in the case of GRFs. As a result, the MH algorithm cannot be
implemented to infer GRFs.

As detailed in the introduction section, a number of variants of the MH algorithm bypass the need to calculate the ratio $Z(\theta)/Z(\theta')$, replacing it in Eq. \eqref{eqn:MHratio} by an unbiased estimator
\begin{equation}
\label{eq:estimatorZ}
\varrho_n(\theta,\theta',x)=\frac{1}{n}\sum_{k=1}^n\frac{q_{\theta}(x_k)}{q_{\theta'}(x_k)}\,,\qquad x_1,x_2,\ldots\sim_\iid f(\,\cdot\,|\,\theta')\,.
\end{equation}
Perhaps surprisingly, when $n=1$ the resulting algorithm, known as the \textit{exchange} algorithm \citep{murray06}, is $\pi$-invariant. The general implementation using $n>1$ auxiliary draws was proposed in \citet{alquier16} and referred therein as the \textit{noisy exchange} algorithm. It is not $\pi$-invariant but the asymptotic bias in distribution was studied in \citep{alquier16}. We note however that when $n$ is large, the resulting algorithm bears little resemblance with the exchange algorithm and really aims at approximating the MH acceptance ratio \eqref{eqn:MHratio}. For clarity, we will therefore refer to the exchange algorithm whenever $n=1$ draw of the likelihood is needed at each iteration and to the noisy Metropolis-Hastings whenever $n>1$.

From Eq. \eqref{eq:estimatorZ}, we see that those modified MH algorithms crucially rely on the ability to sample efficiently from the likelihood distribution ($X\sim f(\,\cdot\,|\,\theta)$ for any $\theta\in\Theta$).
While perfect sampling is possible for certain GRFs, for example for the Ising model \citep{propp96}, it can be computationally expensive in some cases, including large Ising graphs. For some GRFs such as the exponential
random graph model, perfect sampling does not even exist yet.  \citet{cucala09} and \cite{caimo11} substituted the iid sampling in Eq. \eqref{eq:estimatorZ} with $n=1$ draw from a long auxiliary Markov chain that
admits $f(\,\cdot\,|\,\theta)$ as stationary distribution. Convergence of this type of approximate exchange algorithm was studied in \citet{Everitt12} under certain assumptions on the main Markov chain. The
computational bottleneck of those methods is clearly the simulation step, a drawback which is amplified when $n$ is large and inference is on high-dimensional data such as large graphs.

Intuitively, obtaining a likelihood sample at each step independently of the past history of the chain seems to be an inefficient strategy. Indeed, the Markov chain may return to areas of the state space previously visited. As a result, realizations from the likelihood function  are simulated at similar parameter values multiple times, throughout the algorithm. Under general assumptions on the likelihood function, data simulated at similar parameter values will share similar statistical features. Hence, repeated sampling without accounting for previous likelihood simulations seems to lead to an inefficient use of computational time. However, the price to pay to use information from the past history of the chain to speed up the simulation step is the loss of the Markovian dynamic of the chain, leading to a so-called adaptive Markov chain (see \eg \cite{andrieu2008tutorial}). We do not pursue this approach in this paper, essentially since convergence results for adaptive Markov chains depart significantly from the theoretical arguments supporting the validity of the exchange and its variants.

In a different context, \citet{moores14} addressed the computational expense of repeated simulations of Gibbs random fields used within an Approximate Bayesian Computation algorithm (ABC). The authors defined a
pre-processing step designed to learn about the distribution of the summary statistics of simulated data. Part of the total computational budget is spent offline, simulating data from parameter values across
the parameter space $\Theta$. Those pre-simulated data are interpolated to create a mapping function $\Theta\to \Sset$ that is then used during the course of the ABC algorithm to assign an (estimated) sufficient
statistics vector to any parameter $\theta$ for which simulation would be otherwise needed. \citet{moores15} examined a particular GRF, the single parameter hidden Potts model. They combined the pre-processing idea
with path sampling \citep{gelman98} to estimate the ratio of intractable normalising constants. The method presented in \citet{moores15} is suitable for single parameter models but the interpolation step remains a
challenge when the dimension of the parameter space is greater than $1$.

Inspired by the efficiency of a pre-computation step, we develop a novel class of MCMC algorithms, \textit{Pre-computing Metropolis-Hastings}, which uses pre-computed data simulated offline to estimate each normalizing constant ratio $Z(\theta)/ Z(\theta')$ in Eq. \eqref{eqn:MHratio}. This makes the extension to multi-parameter models straightforward.
The steps undertaken during the pre-computing stage are now outlined.

\subsection{Pre-computation step}
\label{sec:prepro}Firstly, a set of parameter values, referred to as a grid, $\Gfrak:=(\tdot_1,...,\tdot_M)$ must be chosen from which to sample graphs from. $\Gfrak$ should cover the full state space and especially the areas of high probability of $\pi$. Finding areas of high probability is not straightforward as this requires knowledge of the posterior distribution. Fortunately, for GRFs we can use Monte Carlo methods to obtain estimates of the gradient and the Hessian matrix of the log posterior at different values of the parameters, which will allow to build a meaningful grid. For a GRF, the well known identity
$$
\nabla_\theta\log \pi(\theta|y)= s(y) - \ex_{f(\,\cdot\,|\,\theta)}s(X) + \nabla_\theta\log p(\theta)
$$
allows the derivation of the following unbiased estimate of the gradient of the log posterior at a parameter $\theta\in\Theta$:
\begin{equation}
\mathcal{G}(\theta,y):= s(y) - \dfrac{1}{N}\sum_{i=1}^N s(X_i) + \nabla_\theta\log p(\theta)\,,\qquad X_1,X_2\,\ldots \sim_\iid f(\,\cdot\,|\,\theta)\label{est_grad}.
\end{equation}
Similarly, the Hessian matrix of the log posterior at a parameter $\theta\in\Theta$ can be unbiasedly estimated by:
\begin{multline}
\mathcal{H}(\theta):=\frac{1}{N-1}\sum_{i=1}^{N} \left\{s(X_i)-\bar{s}\right\}\left\{s(X_i)-\bar{s}\right\}^T+ \nabla^2\log p(\theta)\,,\\
 X_1,X_2\,\ldots \sim_\iid f(\,\cdot\,|\,\theta)\,,\label{est_hes}
\end{multline}
where $\bar{s}$ is the average vector of simulated sufficient statistics.

The grid specification begins by estimating the mode of the posterior $\theta^\ast$. This is achieved by mean of a stochastic approximation algorithm (\eg the Robbins-Monro algorithm \citep{robbins51}), using the log posterior gradient estimate $\mathcal{G}$ defined at Eq. \eqref{est_grad}.

The second step is to estimate the Hessian matrix of the log posterior at $\theta^\ast$ using Eq.~\eqref{est_hes}, in order to get an insight of the posterior curvature at the mode. We denote by $V:=[v_1,\ldots,v_d]$ the matrix whose columns are the eigenvectors $v_i$ of the inverse Hessian at the mode and by $\Lambda:=\text{diag}(\lambda_1,\ldots,\lambda_d)$ the diagonal matrix filled with its eigenvalues. The idea is to construct a grid that preserves the correlations between the variables. It is achieved by taking regular steps in the uncorrelated space \ie the space spanned by $[v_1,\ldots,v_n]$, starting from $\theta^\ast$ and until subsequent estimated gradients are close to each other. The idea is that, for regular models, once the estimated gradients of two successive parameters are similar, the grid has hit the posterior distribution support boundary. Two tuning parameters are required: a threshold parameter for the gradient comparison $m>0$ and an exploratory magnitude parameter $\varepsilon>0$. The grid specification is rigorously outlined in Algorithm \ref{alg:grid}.  Note that in Algorithm \ref{alg:grid}, we have used the notation $\delta_j$ for the $d$-dimensional indicator vector of direction $j$ \ie $\{\delta_j\}_\ell=\1_{j=\ell}$

\begin{algorithm}
\caption{Grid specification}\label{alg:grid}
\begin{algorithmic}[1]
\State \textbf{require}  $\theta^\ast$, $V$, $\Lambda$, $m$ and $\varepsilon$.
\State Initialise the grid with $\Gfrak=\{\theta^\ast\}$
\For{$i\in\{1,\ldots,d\}$}
\ForAll{$\theta\in\Gfrak$}
\State Set $j=0$ and $\theta_0=\theta$
\State Calculate $\tilde{\theta}=\theta_0+\varepsilon V\Lambda^{1/2}\delta_{i}$
\While{$\|\mathcal{G}(\tilde{\theta})-\mathcal{G}(\theta_{j})\|>m$}
\State Set $j=j+1$, $\theta_j=\tilde{\theta}$ and $\Gfrak=\Gfrak\cup \{\theta_j\}$
\State Calculate $\tilde{\theta}=\theta_j+\varepsilon V\Lambda^{1/2}\delta_{i}$
\EndWhile
\EndFor
\EndFor
\State Obtain a second grid $\Gfrak'$ by repeating steps (2)--(12), but moving in the negative direction \ie\, $\tilde{\theta}=\theta-\varepsilon V\Lambda^{1/2} \delta_{i}$.
\State \textbf{return} $\Gfrak=\Gfrak\cup\Gfrak'$
\end{algorithmic}
\end{algorithm}

The left panel of Figure~\ref{fig:pre_grid} shows an example of a naively chosen grid built following standard coordinate directions for a two dimensional posterior distribution. The grid on the right hand side is adapted to the topology of the posterior distribution as described above. This method can be extended to higher dimensional models, but the number of sample grid points would then increase exponentially with dimension. In this paper we do not look beyond two dimensions.


\begin{figure}
\centering
\begin{subfigure}[b]{0.49\linewidth}
\begin{tikzpicture}[scale=1.1]
	\draw (2.4,2.2) circle [x radius=2cm, y radius=5mm, rotate=45];
	\draw (2.4,2.2) circle [x radius=1.5cm, y radius=3.75mm, rotate=45];
  	\draw (2.4,2.2) circle [x radius=1cm, y radius=2.5mm, rotate=45];
  	\draw (2.4,2.2) circle [x radius=0.5cm, y radius=1.25mm, rotate=45];
	
  	\draw[->] (0,0) -- (0,4) node[left] {$\theta_2$};
  	\draw[->] (0,0) -- (5,0) node[below] {$\theta_1$};
	\foreach \x  in {0,...,8}{
	\foreach \y  in {0,...,7}{
		\node[fill=black,draw=black,circle,inner sep=0.5pt] at (\x/2+0.4,\y/2+0.2) {};
		}
		}

\end{tikzpicture}
\end{subfigure}
\begin{subfigure}[b]{0.49\linewidth}
\begin{tikzpicture}[scale=1.1]
	\draw (2.4,2.2) circle [x radius=2cm, y radius=5mm, rotate=45];
	\draw (2.4,2.2) circle [x radius=1.5cm, y radius=3.75mm, rotate=45];
  	\draw (2.4,2.2) circle [x radius=1cm, y radius=2.5mm, rotate=45];
  	\draw (2.4,2.2) circle [x radius=0.5cm, y radius=1.25mm, rotate=45];
  	
	\foreach \x  in {-4,...,12}{
	\foreach \y  in {4,...,10}{
		\node[fill=black,draw=black,circle,inner sep=0.5pt] at (3-\y/5+\x/5,\y/5+\x/5 ) {};
		}
		}
	\node[fill=white,draw=white,circle,inner sep=1pt] at (3-10/5-4/5,10/5-4/5 ) {};
	\node[fill=white,draw=white,circle,inner sep=1pt] at (3-10/5-3/5,10/5-3/5 ) {};
	\node[fill=white,draw=white,circle,inner sep=1pt] at (3-10/5-2/5,10/5-2/5 ) {};
	\node[fill=white,draw=white,circle,inner sep=1pt] at (3-9/5-4/5,9/5-4/5 ) {};
	
	\node[fill=white,draw=white,circle,inner sep=1pt] at (3-5/5+11/5,5/5+11/5 ) {};
	\node[fill=white,draw=white,circle,inner sep=1pt] at (3-5/5-4/5,5/5-4/5 ) {};
	\node[fill=white,draw=white,circle,inner sep=1pt] at (3-5/5-3/5,5/5-3/5 ) {};
	
	\node[fill=white,draw=white,circle,inner sep=1pt] at (3-4/5+12/5,4/5+12/5 ) {};
	\node[fill=white,draw=white,circle,inner sep=1pt] at (3-4/5+11/5,4/5+11/5 ) {};
	\node[fill=white,draw=white,circle,inner sep=1pt] at (3-4/5-4/5,4/5-4/5 ) {};
	\node[fill=white,draw=white,circle,inner sep=1pt] at (3-4/5-3/5,4/5-3/5 ) {};
	\node[fill=white,draw=white,circle,inner sep=1pt] at (3-4/5-2/5,4/5-2/5 ) {};
	\node[fill=white,draw=white,circle,inner sep=1pt] at (3-10/5+10/5,10/5+10/5 ) {};
	\node[fill=white,draw=white,circle,inner sep=1pt] at (3-10/5+11/5,10/5+11/5 ) {};
	\node[fill=white,draw=white,circle,inner sep=1pt] at (3-10/5+12/5,10/5+12/5 ) {};
	\node[fill=white,draw=white,circle,inner sep=1pt] at (3-9/5+12/5,9/5+12/5 ) {};
	
	\node[fill=white,draw=white,circle,inner sep=1pt] at (3-5/5+12/5,5/5+12/5 ) {};
	\node[fill=white,draw=white,circle,inner sep=1pt] at (3-6/5+12/5,6/5+12/5 ) {};

	\draw[->] (0,0) -- (0,4) node[left] {$\theta_2$};
  	\draw[->] (0,0) -- (5,0) node[below] {$\theta_1$};
	
\end{tikzpicture}
\end{subfigure}
\caption{Example of a naive (left panel) and informed (right) grid for a two dimensional posterior distribution. The informed grid was obtained using the process described in Algorithm \ref{alg:grid}.} \label{fig:pre_grid}
\end{figure}
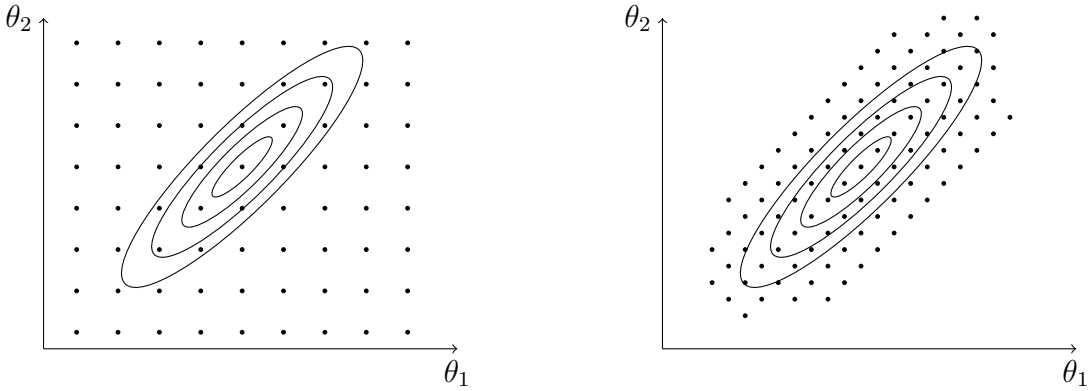

Hereafter, we denote by $\{\tdot_m,\,m\leq M\}$ the parameters constituting the grid $\Gfrak$, assuming $M$ grid points in total. The second step of the pre-computing step is to sample for each $\tdot_m\in\Gfrak$, $n$ \iid random variables $(X_m^1,...,X_m^n)$ from the likelihood function $f(\,\cdot\, | \tdot_m)$. Note that this step is easily parallelised and samples can therefore be obtained from several grid points simultaneously. Parallel processing can be used to reduce considerably the time taken to sample from every pre-computed grid value. Essentially, these draws allow to form unbiased estimators for any ratio of the type $Z(\theta)\slash Z(\tdot_m)$:
\begin{align}
\label{eqn:IS_est}
\widehat{\frac{Z(\theta)}{Z(\tdot_m)}}_n:=\dfrac{1}{n}\displaystyle\sum_{k=1}^n\dfrac{q_{\theta}(X^k_m)}{q_{\tdot_m}(X^k_m)}=\dfrac{1}{n}\displaystyle\sum_{k=1}^n
\exp (\theta-\tdot_m)^Ts(X^k_m)\,.
\end{align}
Note that those estimators depend on the simulated data only through the sufficient statistics $\sfrak_m^k:=s(X_m^k)$. As a consequence, only the sufficient statistics $\Sfrak:=\{\sfrak_m^k\}_{m,k}$ need to be saved, as opposed to the actual collection of simulated graphs at each grid point. In the following we denote by $\Ufrak:=\{\Sfrak,\Gfrak\}$ the collection of the pre-computing data comprising of the grid $\Gfrak$ and the simulated sufficient statistics $\Sfrak$.

\subsection{Estimators of the ratio of normalising constants}

We now detail several pre-computing version of the Metropolis-Hastings algorithm. The central idea is to replace the ratio of normalizing constants
in the Metropolis-Hastings acceptance probability \eqref{eqn:MHratio} by an estimator based on $\Ufrak$. As a starting point this can be done by observing that for all $(\theta,\theta',\tdot)\in\Theta^3$,
\begin{equation}
\label{eq0}
\frac{Z(\theta)}{Z(\theta')}=\frac{Z(\theta)}{Z(\tdot)}\frac{Z(\tdot)}{Z(\theta')}=\frac{Z(\theta)}{Z(\tdot)}\bigg\slash\frac{Z(\theta')}{Z(\tdot)}\,,
\end{equation}
and in particular for any grid point $\tdot\in\Gfrak$. We thus consider a general class of estimators of $Z(\theta)\slash Z(\theta')$ written as
\begin{equation}
\label{eq1}
\rho_n^X(\theta,\theta',\Ufrak):=\frac{\Psi_n^X(\theta,\theta',\Ufrak)}{\Phi_n^X(\theta,\theta',\Ufrak)}\,,
\end{equation}
where $\Psi_n$ and $\Phi_n$ are unbiased estimators of the numerator and the denominator of the right hand side of \eqref{eq0}, respectively, based on $\Ufrak$. In \eqref{eq1}, $X$ simply denotes the different type of estimators considered. To simplify notations and in absence of ambiguity, the dependence of $\rho_n$, $\Psi_n$ and $\Phi_n$ on $\theta$, $\theta'$, $\Ufrak$ and $X$ is made implicit and we stress that given $(\theta,\theta',\Ufrak,X)$, the estimators $\Psi_n$ and $\Phi_n$ are deterministic.

We first note that $\rho_n$ as defined in \eqref{eq1} is not an unbiased estimator of $Z(\theta)/Z(\theta')$. In fact, resorting to biased estimators of the normalizing constants ratio is the price to pay for using
the pre-computed data. This represents a significant departure compared to the algorithms designed in the noisy MCMC literature \citep{alquier16,medina2016stability}. Nevertheless, as we shall see in the next Section,
this does not prevent us from controlling the distance between the distribution of the pre-computing Markov chain and $\pi$.

We propose a number of different estimators of $\Psi_n$ and $\Phi_n$. Those estimators share in common the idea that, given the current chain location $\theta$ and an attempted move $\theta'$, a path of grid point(s) $\{\tdot_{\tau_1},\tdot_{\tau_2},\ldots,\tdot_{\tau_C}\}\subset \Gfrak$ connects $\theta$ to $\theta'$.

The simplest path consists of the singleton $\{\tdot_{\tau}\}$, where $\tdot_{\tau}$ is any grid point. Since only one grid point is used, we refer to this estimator as the \textit{One Pivot} estimator.
Following (\ref{eq0}), the estimators $\Psi_n$ and $\Phi_n$ are defined as
\begin{equation}
\label{eq:1p:1}
\left\{
\begin{array}{l}
\Psi_n^{\OP}(\theta,\theta',\Ufrak):=1/n\sum_{k=1}^n {q_{\theta}(X_{\tau}^k)}\slash{q_{\tdot_{\tau}}(X_{\tau}^k)}\,,\\
\Phi_n^{\OP}(\theta,\theta',\Ufrak):=1/n\sum_{k=1}^n {q_{\theta'}(X_\tau^k)}\slash{q_{\tdot_\tau}(X_{\tau}^k)}\,.\\
\end{array}
\right.
\end{equation}
However, for some $(\theta,\theta',\tdot_\tau)\in\Theta^2\times\Gfrak$, the variance of $\Psi_n$ or $\Phi_n$ defined in Eq. \eqref{eq:1p:1} may be large. This is especially likely when $\|\theta-\tdot_\tau\|\gg 1$ or $\|\theta'-\tdot_\tau\|\gg 1$.  The following Example illustrates this situation.

\begin{exa}
\label{ex:erdos}
Consider the Erd\"{o}s-Renyi graph model, where all graphs $y\in\Yset$ with the same number of edges $s(y)$ are equally likely. More precisely, the dyads are independent and connected with a probability $\varrho(\theta):=\text{logit}^{-1}(\theta)$ for any $\theta\in\rset$. The likelihood function is given for any $\theta\in\rset$ by $f(y\,|\,\theta)\propto \exp\{\theta s(y)\}$. For this model, the normalizing constant is tractable. In particular, $Z(\theta)=\{1+\exp(\theta)\}^{\barp}$ where $\barp={p\choose 2}$ and $p$ is the number of nodes in the graph.

For all $\theta\in\rset$, consider estimating the ratio $Z(\theta')\slash Z(\theta)$ with $\theta'=\theta+h$ for some $h>0$ using the estimator
$$
\left.\widehat{\frac{Z(\theta+h)}{Z(\theta)}}\right|_n=\frac{1}{n}\sum_{k=1}^n\frac{q_{\theta+h}(X_k)}{q_\theta(X_k)}=\frac{1}{n}\sum_{k=1}^n\exp\{h s(X_k)\}\,,\qquad X_k\sim_\iid f(\,\cdot\,|\,\theta)\,.
$$
Then, when $h$ increases, the variance $v_n$ of this estimator diverges exponentially \ie
\begin{equation}
\label{eq:erdos}
n v_n(h)\sim  \exp(2h\barp)\nu(\theta)\,,
\end{equation}
where $\sim$ denotes here the asymptotic equivalence notation and $\nu(\theta)=\varrho(\theta)^{\barp}(1-\varrho(\theta)^{\barp})$ is a constant. Remarkably, $\nu(\theta)$ can be interpreted as the variance of the Bernoulli trial with the full graph and its complementary event as outcomes.
\end{exa}

\begin{proof}
By straightforward algebra, we have
$$
v_n(h)=\frac{1}{n}\left\{\frac{1+\exp(2h+\theta)}{1+\exp(\theta)}\right\}^{\barp}\left\{1-R(\theta,h)\right\}\,,
$$
where
$$
R(\theta,h)=\frac{\{1+\exp(\theta+h)\}^{2\barp}}{\{1+\exp(2h+\theta)\}^{\barp}\left\{1+\exp(\theta)\right\}^{\barp}}\,.
$$
Asymptotically in $h$, we have
$$
R(\theta,h)\sim \frac{\exp(\barp\theta)}{\left\{1+\exp(\theta)\right\}^{\barp}}=\varrho(\theta)^{\barp}
$$
and noting that
$$
\left\{\frac{1+\exp(2h+\theta)}{1+\exp(\theta)}\right\}^{\barp}\sim \exp(2h\barp)\frac{\exp\{\barp\theta\}}{\{1+\exp(\theta)\}^{\barp}}=\exp(2h\barp)\varrho(\theta)^{\barp}
$$
concludes the proof.
\end{proof}

This is a concern since as we shall see in the next Section, the noise introduced by the pre-computing step in the Markov chain is intimately related to the variance of the estimator of $Z(\theta)\slash Z(\theta')$. In particular, the distance between the pre-computing chain distribution and $\pi$ can only be controlled when the variance of $\Psi_n$ and $\Phi_n$ is bounded. Example \ref{ex:erdos} shows that this is not necessarily the case, for some Gibbs random fields at least. The following Proposition hints at the possibility to control the variance of $\Psi_n$ and $\Phi_n$ when $\|\theta-\theta'\|\ll 1$.

\begin{prop}
\label{prop1}
For any Gibbs random field model and all $(\theta,\theta')\in\Theta^2$, the variance of the normalizing constant estimator
$$
\left.\widehat{\frac{Z(\theta)}{Z(\theta')}}\right|_n:=\frac{1}{n}\sum_{k=1}^n\frac{q_\theta(X_k)}{q_{\theta'}(X_k)}\,,\qquad X_k\sim_\iid f(\,\cdot\,|\,\theta')
$$
decreases when $\|\theta-\theta'\|\downarrow 0$ and more precisely
\begin{equation}
\label{eq:prop1}
\var\left.\widehat{\frac{Z(\theta)}{Z(\theta')}}\right|_n=\mathcal{O}(\|\theta-\theta'\|^2)\,.
\end{equation}
\end{prop}

Proposition \ref{prop1} motivates the consideration of estimators that may have smaller variability than the One Pivot estimator. 
\begin{enumerate}[(1)]
\item \textit{Direct Path} estimator: the path between $\theta$ and $\theta'$ consists now of two grid points $\{\tdot_1,\tdot_2\}$ defined such that $\tdot_1=\arg\min_{\tdot\in \Gfrak}\|\tdot-\theta\|$
and $\tdot_2=\arg\min_{\tdot\in \Gfrak}\|\tdot-\theta'\|$. We therefore extend \eqref{eq0} and write
\[
 \frac{Z(\theta)}{Z(\theta')}=\frac{Z(\theta)}{Z(\tdot_1)}\frac{Z(\tdot_1)}{Z(\tdot_2)}\frac{Z(\tdot_2)}{Z(\theta')}=\frac{Z(\theta)}{Z(\tdot_1)}\frac{Z(\tdot_1)}{Z(\tdot_2)}\bigg\slash\frac{Z(\theta')}{Z(\tdot_2)}.
\]
This leads to two estimators $\Psi_n$ and $\Phi_n$ defined as
    \begin{equation}
\label{eq:dpath}
\left\{
\begin{array}{l}
\Psi_n^{\DP}(\theta,\theta',\Ufrak):=1/n\sum_{k=1}^n {q_{\theta}(X_1^k)}\slash{q_{\tdot_1}(X_1^k)}\times 1/n\sum_{k=1}^n {q_{\tdot_1}(X_k^2)}\slash{q_{\tdot_2}(X_k^2)}\,,\\
\Phi_n^{\DP}(\theta,\theta',\Ufrak):=1/n\sum_{k=1}^n {q_{\theta'}(X_2^k)}\slash{q_{\tdot_2}(X_2^k)}\,.\\
\end{array}
\right.
\end{equation}
\item \textit{Full Path} estimator: the path between $\theta$ and $\theta'$ consists now of adjacent grid points $\Pfrak(\theta,\theta'):=\{\tdot_1,\tdot_2,\ldots,\tdot_C\}$, where $C>1$ is a number that depends on $\theta$ and $\theta'$. Note that given $(\theta,\theta')$, there is not only one path such as $\Pfrak$ connecting $\theta$ to $\theta'$. However, for any possible path, two adjacent points $\{\tdot_i,\tdot_{i+1}\}\subset\Pfrak(\theta,\theta')$ always satisfy the following identity (in the basis given by the eigenvector of $\mathcal{H}(\theta^\ast)$):
    $$
    \exists\,j\in\{1,\ldots,d\}\,,\qquad V\T\left(\tdot_i-\tdot_{i+1}\right)=\pm \varepsilon \delta_j\,,
    $$
    where $\delta_j$ refers to the $d$-dimensional indicator vector of direction $j$ \ie $\{\delta_j\}_\ell=\1_{j=\ell}$. As before, we extend \eqref{eq0} to accommodate this situation and write
    \[
 \frac{Z(\theta)}{Z(\theta')}=\frac{Z(\theta)}{Z(\tdot_1)}\frac{Z(\tdot_1)}{Z(\tdot_2)}\times \dots\times \frac{Z(\tdot_{C-1})}{Z(\tdot_C)}\frac{Z(\tdot_C)}{Z(\theta')}=\frac{Z(\theta)}{Z(\tdot_1)}\frac{Z(\tdot_1)}{Z(\tdot_2)}\times \dots\times \frac{Z(\tdot_{C-1})}{Z(\tdot_C)}\bigg\slash\frac{Z(\theta')}{Z(\tdot_c)}.
\]
   This then lead to consider two estimators $\Psi_n$ and $\Phi_n$ defined as
\begin{equation}
\label{eq:fpath}
\left\{
\begin{array}{l}
\Psi_n^{\FP}(\theta,\theta',\Ufrak):=1/n\sum_{k=1}^n {q_{\theta}(X_1^k)}\slash{q_{\tdot_1}(X_1^k)}\times 1/n\sum_{k=1}^n {q_{\tdot_1}(X_k^2)}\slash{q_{\tdot_2}(X_k^2)}\\
\hfill \times\cdots\times 1/n\sum_{k=1}^n {q_{\tdot_{C-1}}(X_{C-1}^k)}\slash{q_{\tdot_C}(X_C^k)}\,,\\
\Phi_n^{\FP}(\theta,\theta',\Ufrak):=1/n\sum_{k=1}^n {q_{\theta'}(X_C^k)}\slash{q_{\tdot_C}(X_C^k)}\,.\\
\end{array}
\right.
\end{equation}
\end{enumerate}

Variants of the Direct Path and Full Path estimators exist. For the Direct Path, $\Psi_n$ could be estimating $Z(\theta)/Z(\tdot_{\tau_1})$ and $\Phi_n$ the ratio $Z(\tdot_{\theta'})/Z(\theta_{\tau_1})$. For
the Full Path, defining $\tdot_{\tau_m}$ as a middle point of $\Pfrak(\theta,\theta')$, $\Phi_n$ and $\Psi_n$ could respectively be defined as estimators of $Z(\theta)/Z(\tdot_{\tau_m})$ and $Z(\theta')/Z(\tdot_{\tau_m})$
using the same number of grid points in both estimators. However, our experiments have shown that these alternative estimators have very similar behaviour with those defined in Eqs. \eqref{eq:dpath} and \eqref{eq:fpath}.
In particular, the variance of an estimator does not vary much when path points are removed from the numerator estimator and added to the denominator estimator, or conversely. As hinted by Proposition \ref{prop1}, the
discriminant feature between those estimators is the distance between grid points constituting the path. In this respect, the variance of the Full Path estimator was always found to be lower than that of the Direct Path or One Pivot estimators. Even though establishing a rigorous comparison result between those estimators is a challenge on its own, a reader might be interested in the following result that somewhat formalizes our empirical observations.

\begin{prop}
\label{prop2}
Let $(\theta,\theta')\in\Theta$ and consider the Direct Path and Full Path estimators of $Z(\theta)/Z(\theta')$ defined at \eqref{eq:dpath} and \eqref{eq:fpath}. Denoting by $\{\tdot_1,\ldots,\tdot_C\}$ a full path connecting $\theta$ to $\theta'$, we define for $i\in\{2,\ldots,C\}$ $R_n^i$ as the estimator of $Z(\tdot_{i-1})/Z(\tdot_{i})$ and $R_n^{2C}$ as the estimator of $Z(\tdot_1)/Z(\tdot_C)$ \ie
\begin{equation}
R_n^i=\frac{1}{n}\sum_{k=1}^n\frac{q_{\tdot_{i-1}}(X_i^k)}{q_{\tdot_i}(X_i^k)}\,,\qquad R_n^{2C}=\frac{1}{n}\sum_{k=1}^n\frac{q_{\tdot_{1}}(X_C^k)}{q_{\tdot_C}(X_C^k)}\,,\qquad
X_i^k\sim_{\iid} f(\,\cdot\,|\,\tdot_i)\,.
\end{equation}
Let $v_n^{\FP}$ and $v_n^{\DP}$ be the variance of the Full Path and Direct Path estimators using $n$ pre-computed sufficient statistics are drawn at each grid point.

Assume $\Phi_n$ and $\Psi_n$ are independent. Then, there exists a positive constant $\gamma<\infty$ such that
\begin{equation}
\label{eq:prop2:1}
v_n^{\DP}-v_n^{\FP}=\gamma\left\{\var(R_n^{2C})-\var(R_n^2\times\cdots\times R_n^C)\right\}\,.
\end{equation}
Moreover,
\begin{equation}
\var(R_n^{2C})=\frac{1}{n}\var\exp\left\{(\tdot_1-\tdot_C)\T s(X_C)\right\}
\end{equation}
and for large $n$ and $C$ and small $\eps$ we have
\begin{equation}
\label{eq:prop2:2}
\var(R_n^2\times\cdots\times R_n^C)=\frac{\eps^4}{n}\sum_{i=2}^C\left\{v_i\frac{Z(\tdot_i)Z(\tdot_1)}{Z(\tdot_{i-1})Z(\tdot_C)}\right\}^2+o\left(\eps^4/n\right)\,,
\end{equation}
where $\{v_1,v_2,\ldots\}$ is a sequence of finite numbers such that $v_i\in\mathcal{O}(\eps)$.
\end{prop}

Proposition \ref{prop2} shows that for a large enough number of pre-computed draws $n$, a long enough path and a dense grid \ie $\epsilon\ll 1$, the variance of the Full Path estimator is several order of
magnitude less than that of the Direct Path estimator. In particular, unlike the Full Path estimator, the grid refinement does not help to reduce the variance of the Direct Path estimator. Proposition \ref{prop2}
coupled with the observation made at Example \ref{ex:erdos} helps to understand the variance reduction achieved with the Full Path estimator compared to the Direct Path estimator.

Note that when the parameter space is two-dimensional or higher, there is more than one choice of path connecting $\theta$ to $\theta'$. The right panel of Figure~\ref{fig:pre_eg} shows two different paths. In
this situation, one could simply average the Full Path estimators obtained through each (or a number of) possible path. The different steps included in the Pre-computing Metropolis algorithm are summarized in
Algorithm \ref{alg:NoisyExch}.

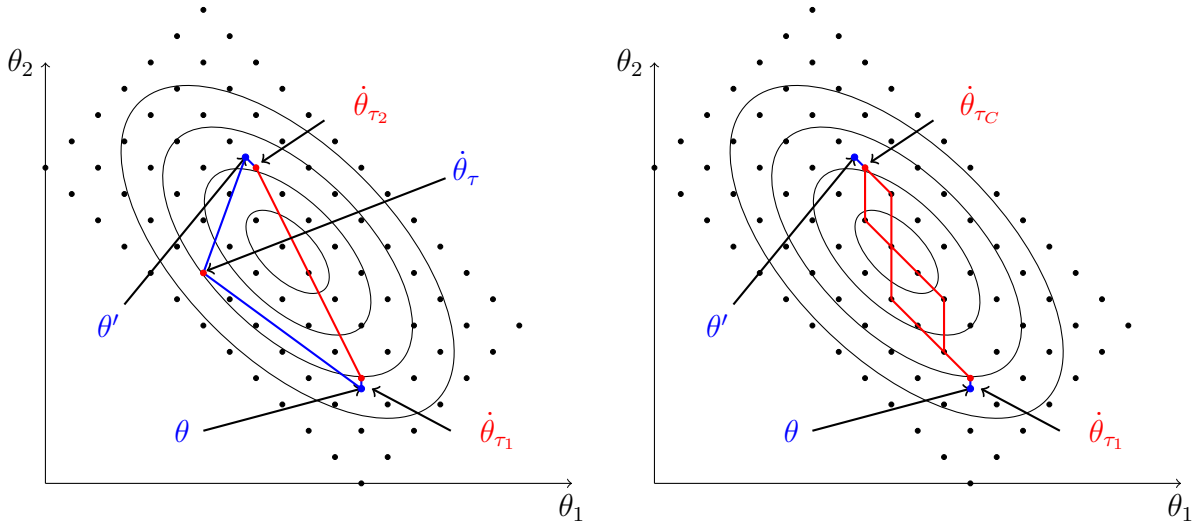
\begin{figure}
\centering
\begin{tikzpicture}[scale=1.4]
	\draw (2.3,2.2) circle [x radius=2cm, y radius=10mm, rotate=-45];
	\draw (2.3,2.2) circle [x radius=1.5cm, y radius=7.5mm, rotate=-45];
  	\draw (2.3,2.2) circle [x radius=1cm, y radius=5mm, rotate=-45];
  	\draw (2.3,2.2) circle [x radius=0.5cm, y radius=2.5mm, rotate=-45];
  	\draw[->] (0,0) -- (0,4) node[left] {$\theta_2$};
  	\draw[->] (0,0) -- (5,0) node[below] {$\theta_1$};
	\foreach \x  in {0,...,6}{
		\node[fill=black,draw=black,circle,inner sep=0.65pt] at (0+\x/4 ,3+\x/4) {};
		\node[fill=black,draw=black,circle,inner sep=0.65pt] at (0.25+\x/4 ,2.75+\x/4) {};
		\node[fill=black,draw=black,circle,inner sep=0.65pt] at (0.5+\x/4 ,2.5+\x/4) {};
		\node[fill=black,draw=black,circle,inner sep=0.65pt] at (0.75+\x/4 ,2.25+\x/4) {};
		\node[fill=black,draw=black,circle,inner sep=0.65pt] at (1+\x/4 ,2+\x/4) {};
		\node[fill=black,draw=black,circle,inner sep=0.65pt] at (1.25+\x/4 ,1.75+\x/4) {};
		\node[fill=black,draw=black,circle,inner sep=0.65pt] at (1.5+\x/4 ,1.5+\x/4) {};
		\node[fill=black,draw=black,circle,inner sep=0.65pt] at (1.75+\x/4 ,1.25+\x/4) {};
		\node[fill=black,draw=black,circle,inner sep=0.65pt] at (2+\x/4 ,1+\x/4) {};
		\node[fill=black,draw=black,circle,inner sep=0.65pt] at (2.25+\x/4 ,0.75+\x/4) {};
		\node[fill=black,draw=black,circle,inner sep=0.65pt] at (2.5+\x/4 ,0.5+\x/4) {};
		\node[fill=black,draw=black,circle,inner sep=0.65pt] at (2.75+\x/4 ,0.25+\x/4) {};
		\node[fill=black,draw=black,circle,inner sep=0.65pt] at (3+\x/4 ,0+\x/4) {};}
		
	\draw[draw=blue,thick] (3,1) -- (3,0.9);
	\draw[draw=blue,thick] (2,3) -- (1.9,3.1);
	\node[fill=red,draw=red,circle,inner sep=0.75pt] at (3,1){};
	\node[fill=red,draw=red,circle,inner sep=0.75pt] at (2,3){};
	\draw[draw=red,thick] (2,3) -- (3,1);

	\node at (4.3,0.5) {\textcolor{red}{$\tdot_{\tau_1}$}};	
	\node at (3.1,3.6) {\textcolor{red}{$\tdot_{\tau_2}$}};

	\draw[->,thick] (3.85,0.5) -- (3.1,0.9){};
	\draw[->,thick] (2.65,3.45) -- (2.05,3.05){};
	
	\node at (1.3,0.5) {\textcolor{blue}{$\theta$}};	
	\node at (0.6,1.5) {\textcolor{blue}{$\theta'$}};	
	\draw[->,thick] (1.5,0.5) -- (3,0.9){};
	\draw[->,thick] (0.75,1.7) -- (1.9,3.1){};
	\node[fill=blue,draw=blue,circle,inner sep=0.85pt] at (3,0.9){};
	\node[fill=blue,draw=blue,circle,inner sep=0.85pt] at (1.9,3.1){};


\draw[draw=blue,thick] (1.5,2) -- (3,0.9);
	\draw[draw=blue,thick] (1.5,2) -- (1.9,3.1);
	\node[fill=red,draw=red,circle,inner sep=0.75pt] at (1.5,2){};
	\node at (4,3) {\textcolor{blue}{$\tdot_{\tau}$}};	
		\draw[->,thick] (3.8,2.9) -- (1.53,2.02){};
	
\end{tikzpicture}
\begin{tikzpicture}[scale=1.4]
	\draw (2.3,2.2) circle [x radius=2cm, y radius=10mm, rotate=-45];
	\draw (2.3,2.2) circle [x radius=1.5cm, y radius=7.5mm, rotate=-45];
  	\draw (2.3,2.2) circle [x radius=1cm, y radius=5mm, rotate=-45];
  	\draw (2.3,2.2) circle [x radius=0.5cm, y radius=2.5mm, rotate=-45];
  	\draw[->] (0,0) -- (0,4) node[left] {$\theta_2$};
  	\draw[->] (0,0) -- (5,0) node[below] {$\theta_1$};
	\foreach \x  in {0,...,6}{
		\node[fill=black,draw=black,circle,inner sep=0.65pt] at (0+\x/4 ,3+\x/4) {};
		\node[fill=black,draw=black,circle,inner sep=0.65pt] at (0.25+\x/4 ,2.75+\x/4) {};
		\node[fill=black,draw=black,circle,inner sep=0.65pt] at (0.5+\x/4 ,2.5+\x/4) {};
		\node[fill=black,draw=black,circle,inner sep=0.65pt] at (0.75+\x/4 ,2.25+\x/4) {};
		\node[fill=black,draw=black,circle,inner sep=0.65pt] at (1+\x/4 ,2+\x/4) {};
		\node[fill=black,draw=black,circle,inner sep=0.65pt] at (1.25+\x/4 ,1.75+\x/4) {};
		\node[fill=black,draw=black,circle,inner sep=0.65pt] at (1.5+\x/4 ,1.5+\x/4) {};
		\node[fill=black,draw=black,circle,inner sep=0.65pt] at (1.75+\x/4 ,1.25+\x/4) {};
		\node[fill=black,draw=black,circle,inner sep=0.65pt] at (2+\x/4 ,1+\x/4) {};
		\node[fill=black,draw=black,circle,inner sep=0.65pt] at (2.25+\x/4 ,0.75+\x/4) {};
		\node[fill=black,draw=black,circle,inner sep=0.65pt] at (2.5+\x/4 ,0.5+\x/4) {};
		\node[fill=black,draw=black,circle,inner sep=0.65pt] at (2.75+\x/4 ,0.25+\x/4) {};
		\node[fill=black,draw=black,circle,inner sep=0.65pt] at (3+\x/4 ,0+\x/4) {};}
		
	\draw[draw=blue,thick] (3,1) -- (3,0.9);
	\draw[draw=blue,thick] (2,3) -- (1.9,3.1);
	\node[fill=red,draw=red,circle,inner sep=0.75pt] at (3,1){};
	\node[fill=red,draw=red,circle,inner sep=0.75pt] at (2,3){};

	\draw[draw=red,thick] (2,3) -- (2,2.5);
	\draw[draw=red,thick] (2,2.5) -- (2.25,2.25);
	\draw[draw=red,thick] (2.25,2.25) -- (2.5,2);
	\draw[draw=red,thick] (2.5,2) -- (2.75,1.75);
	\draw[draw=red,thick] (2.75,1.75) -- (2.75,1.25);
	\draw[draw=red,thick] (2.75,1.25) -- (3,1);
	
	\draw[draw=red,thick] (2,3) -- (2.25,2.75);
	\draw[draw=red,thick] (2.25,2.75) -- (2.25,1.75);
	\draw[draw=red,thick] (2.25,1.75) -- (2.75,1.25);

	\node at (4.3,0.5) {\textcolor{red}{$\tdot_{\tau_1}$}};	
	\node at (3.1,3.6) {\textcolor{red}{$\tdot_{\tau_C}$}};	
	\draw[->,thick] (3.85,0.5) -- (3.1,0.9){};
	\draw[->,thick] (2.65,3.45) -- (2.05,3.05){};

	\node at (1.3,0.5) {\textcolor{blue}{$\theta$}};	
	\node at (0.6,1.5) {\textcolor{blue}{$\theta'$}};	
	\draw[->,thick] (1.5,0.5) -- (3,0.9){};
	\draw[->,thick] (0.75,1.7) -- (1.9,3.1){};
	\node[fill=blue,draw=blue,circle,inner sep=0.85pt] at (3,0.9){};
	\node[fill=blue,draw=blue,circle,inner sep=0.85pt] at (1.9,3.1){};
\end{tikzpicture}
\caption{Example of paths between two parameters $(\theta,\theta')$ in a two-dimensional space $\Theta$. The solid black lines represent level lines of the target distribution and the black dots represent the grid vertices $\Gfrak=\{\tdot_1,\ldots,\tdot_M\}$. The thick lines show the paths $\Pfrak(\theta,\theta')$ used by the different estimators introduced in Eqs. \eqref{eq:1p:1}, \eqref{eq:dpath}, \eqref{eq:fpath}: an example of a One Pivot path $\Pfrak(\theta,\theta')=\{\tdot\}$ (blue) and a Direct Path $\Pfrak(\theta,\theta')=\{\tdot_1,\tdot_2\}$ (red) are shown on the left panel. Two examples of Full Paths $\Pfrak(\theta,\theta')=\{\tdot_1,\ldots,\tdot_C\}$ (red) are illustrated on the right panel: multiple possible full paths between $\theta$ and $\theta'$ could be used to average a number of Full Path estimators.} \label{fig:pre_eg}
\end{figure}

\begin{algorithm}
\caption{Pre-computing Metropolis algorithm}
\label{alg:NoisyExch}
{\bf (1)-Pre-computing}
\begin{algorithmic}[1]
\Require Grid refinement parameter $\eps>0$ and number of draws $n\in\nset$
	\State Apply Algorithm~\ref{alg:grid} to define the grid $\Gfrak=\{\tdot_1,\ldots,\tdot_M\}$.
    \State Initiate the collection of sufficient statistics to $\Sfrak=\{\emptyset\}$.
	\For {$m=1$ to $M$}
    \For {$k=1$ to $n$}
		\State Draw $X_{m}^k\sim_\iid f(\cdot\,|\,\tdot_j)$
        \State Calculate the vector of sufficient statistics $\sfrak_m^k=s(X_m^k)$
        \State Append the pre-computed sufficient statistics set $\Sfrak=\{\Sfrak \cup \sfrak_m^k\}$
    \EndFor
	\EndFor\vspace{0.2cm}

\vspace{-.1cm}
\hspace{-1.5cm}\textbf{Return:} The pre-computed data $\Ufrak=\{\Gfrak,\Sfrak\}$
\end{algorithmic}
{\bf (2)-MCMC sampling}
\begin{algorithmic}[1]
\Require Initial distribution $\mu$ and proposal kernel $h$, pre-computed data $\Ufrak$ and a type of estimator $\rho_n^{X}$, $X\in\{\OP,\DP,\FP\}$
\State Initiate the Markov chain with $\theta_0\sim \mu$
\State Identify the closest grid point from $\theta_0$, say $\tdot_i$, and calculate
$$
Z_0:=\dfrac{1}{n}\displaystyle\sum_{k=1}^n\exp\left\{(\theta_0-\tdot_i)^T \sfrak^k_{i}\right\}\,.
$$
	\For {$i=1,2,\ldots$}\vspace{0.1cm}
		\State Draw $\theta'\sim h(\,\cdot\,|\,\theta_{i-1})$
		\State Identify the closest grid point from $\theta'$, say $\tdot_i$, and calculate
\begin{equation*}
Z':=\dfrac{1}{n}\displaystyle\sum_{k=1}^n\exp\left\{(\theta'-\tdot_i)^T \sfrak^k_{i}\right\}\,.
\end{equation*}

		\State Using $Z_{i-1}$, $Z'$ and $\Sfrak$, calculate the normalizing ratio estimator $\rho_n^X$, depending on the type of estimator $X$ using Eq. \eqref{eq:1p:1}, \eqref{eq:dpath} or \eqref{eq:fpath}.
		\State Set $\theta_{i} =\theta'$ and $Z_i=Z'$ with probability \vspace{0.2cm}
		\begin{multline} \bar{\alpha}(\theta_{i-1},\theta',\Ufrak):=1\wedge \bar{a}(\theta_{i-1},\theta',\Ufrak)\,,\\
\bar{a}(\theta_{i-1},\theta',\Ufrak)=\dfrac{q_{\theta'}(y)p(\theta')h(\theta_{i}|\theta')}{q_{\theta_{i}}(y)p(\theta_{i})h(\theta'|\theta_{i})}\times\rho_n^X(\theta_{i-1},\theta',\Ufrak)
\label{eq:nmh}
		\end{multline}
		 and else set $\theta_{i} =\theta_{i-1}$ and $Z_i=Z_{i-1}$.
	\EndFor

\vspace{.1cm}
\hspace{-1.5cm}\textbf{Return:} The Markov chain $\{\theta_1,\theta_2,\ldots\}$.

\end{algorithmic}
\end{algorithm}

\section{Asymptotic analysis of the pre-computing \allowbreak Metropolis-Hastings algorithms}
\label{sec:theory}
In this section, we investigate the theoretical guarantees for the convergence of the Markov chain $\{\theta_k,\,k\in\nset\}$ produced by the pre-computing Metropolis algorithm (Alg. \ref{alg:NoisyExch}) to the posterior distribution $\pi$. The Markov transition kernels considered in this section are conditional probability distributions on the measurable space $(\Theta,\vartheta)$ where $\vartheta$ is the $\sigma$-algebra taken as the Borel set on $\Theta$. We will use the following transition kernels:
\begin{itemize}
\item Let $P$ be the Metropolis-Hastings (MH) transition kernel defined as:
\begin{multline}
\label{eq:P_MH}
P(\theta,A)=\int_A h(\rmd\theta'\,|\,\theta)\alpha(\theta,\theta')+\delta_\theta(A)r(\theta)\,,\quad \\r(\theta)=1-\int_\Theta h(\rmd\theta'\,|\,\theta)\alpha(\theta,\theta')\,,
\end{multline}
where $\delta_\theta$ is the dirac mass at $\theta$ and $\alpha$ the (intractable) MH acceptance probability defined at Eq. \eqref{eqn:MHratio}.
\item Let $\bP_\Ufrak$ be the pre-computing Metropolis transition kernel, conditioned on the pre-computing data $\Ufrak$ and defined as:
\begin{multline}
\label{eq:P_PCEX}
    \bP_\Ufrak(\theta,A)=\int_A h(\rmd\theta'\,|\,\theta)\balpha(\theta,\theta',\Ufrak)+\delta_\theta(A)\bar{r}(\theta,\Sfrak)\,,\quad\\ \bar{r}(\theta,\Ufrak)=1-\int_\Theta h(\rmd\theta'\,|\,\theta)\bar{\alpha}(\theta,\theta',\Ufrak)\,,
\end{multline}
    where $\balpha$ is the pre-computing Metropolis acceptance probability defined at Eq. \eqref{eq:nmh}.
\end{itemize}

We recall that the MH Markov chain is $\pi$-invariant, a property which is lost by the pre-computing Metropolis algorithm. In what follows, we regard $\bP$ as a noisy version of the MH kernel $P$ and $\balpha$ as an approximation of the intractable quantity $\alpha$. In terms of notations, we will use the following: for any $i\in\nset$, $P^i$ is the transition kernel $P$ iterated $i$ times and for any measure $\mu$ on $(\Theta,\vartheta)$, $\mu P$ is the probability measure on $(\Theta,\vartheta)$ defined as $\mu P(A):=\int \mu(\rmd\theta)P(\theta,A)$.

Using the theoretical framework, developed in \citet{alquier16}, we show that under certain assumptions, the distance between the distribution of the pre-computing Metropolis Markov chain and $\pi$ can be made arbitrarily small, in function of the grid refinement and the number of auxiliary draws. The metric used on the space of probability distributions is the total variation distance, defined for two distributions $(\nu,\mu)$ that admit a density function with respect to the Lebesgue measure as
$$
\|\nu-\mu\|:=(1/2)\int_\Theta |\nu(\theta)-\mu(\theta)|\rmd\theta\,.
$$

\subsection{Noisy Metropolis-Hastings}
We first recall the main result from \cite{alquier16} that will be used to analyse the pre-computing Metropolis algorithm.

\begin{prop}[Corollary 2.3 in \citet{alquier16}]
\label{cor1}

\hspace{40cm} Let Let us assume that,
\begin{itemize}
\item ${\bf (H1)}$ A MH Markov chain with transition kernel $P$  (Eq. \ref{eq:P_MH}), proposal kernel $h$ and acceptance probability $\alpha$ (Eq. \ref{eqn:MHratio}) is uniformly ergodic \ie there are constants $B>0$ and $\rho<1$ such that
    $$
    \forall\,i\in\nset\,,\qquad \sup_{\theta_0\in\Theta}\|\delta_{\theta_0}P^i-\pi\|\leq B\rho^i\,.
    $$
\item ${\bf (H2)}$ There exists an approximation of the Metropolis acceptance ratio $a$, $\hat{a}(\theta,\theta',X)$ that satisfies for all $(\theta,\theta')\in\Theta^2$
$$
\mathbb{E}\left|
\hat{a}(\theta,\theta',X)-a(\theta,\theta')\right|\leq \epsilon(\theta,\theta')\,,
$$
where the expectation is with respect to the noise random variable $X$.
\end{itemize}
Then, denoting by $\hat{P}$ the noisy Metropolis-Hastings kernel (Eq. \ref{eq:P_PCEX}), we have for any starting point $\theta_0\in\Theta$ and any integer $i\in\nset$:
\begin{equation}
\label{eq:cor1}
\|\delta_{\theta_0}P^i-\delta_{\theta_0}\hat{P}^i\|\leq\left(\lambda-\dfrac{B\rho^\lambda}{1-\rho}\right)\sup_{\theta\in\Theta}\int d\theta'h(\theta'|\theta)\epsilon(\theta,\theta')\,,
\end{equation}
where $\lambda=\left(\dfrac{\log(1/B)}{\log(\rho)}\right)$.\vspace{3mm}\\
\end{prop}

An immediate consequence of Proposition \ref{cor1} is that if $\epsilon$ is uniformly bounded, \ie there exists some $\bar{\epsilon}>0$ such that for all $(\theta,\theta')\in\Theta^2$, $\epsilon(\theta,\theta')\leq\bar{\epsilon}<\infty$, then
\begin{equation}
\label{eq:cor2}
 \forall\,i\in\nset\,,\qquad \|\delta_{\theta_0}P^i-\delta_{\theta_0}\hat{P}^i\|\leq\bar{\epsilon}\left(\lambda-\dfrac{B\rho^\lambda}{1-\rho}\right)\,.
\end{equation}
Moreover, defining $\hat{\pi}_i$ as the distribution of the $i$-th state of the noisy chain yields
\begin{equation}
\label{eq:cor3}
\lim_{i\to\infty}\|\pi-\hat{\pi}_i\|\leq \bar{\epsilon}\left(\lambda-\dfrac{B\rho^\lambda}{1-\rho}\right)\,.
\end{equation}


\subsection{Convergence of the pre-computing Metropolis algorithm}

In preparation to apply Proposition \ref{cor1}, we make the following assumptions:
\begin{itemize}
\item ${\bf (A1)}$ there is a constant $c_p$ such that for all $\theta\in\Theta$, $1/c_p\leq p(\theta)\leq c_p$.
\item ${\bf (A2)}$ there is a constant $c_h$ such that for all $(\theta,\theta')\in\Theta^2$, $1/c_h\leq h(\theta'|\theta)\leq c_h$.
\end{itemize}
Assumptions ${\bf (A1)}$ and ${\bf (A2)}$ are typically satisfied when $\Theta$ is a bounded set and $p$ and $h(\,\cdot\,|\,\theta)$ are dominated by the Lebesgue measure. Under similar assumptions, Proposition \ref{cor1} was applied to the noisy Metropolis algorithm \citep{alquier16} that uses the unbiased estimator $\varrho_n$ (Eq. \ref{eq:estimatorZ}). More precisely, it was shown that the distance between $\pi$ and $\hpi_i$ satisfies $\|\pi-\hpi_i\|\leq \kappa/\sqrt{n}$, where $\kappa>0$ is a positive constant, asymptotically in $i$.

Establishing an equivalent result for the pre-computing Metropolis algorithms is not straightforward. The main difficulty is that the acceptance ratio $\tilde{a}(\theta,\theta',\Ufrak)$ (Eq. \ref{eq:nmh}) is a biased
estimator of the MH acceptance ratio $a(\theta,\theta')$ (Eq. \ref{eqn:MHratio}). The following Proposition only applies to the pre-computing Metropolis algorithm involving the approximation of the normalizing constant
ratio using the full path estimator. Weaker results can be obtained using similar arguments for the One Pivot and Direct Path estimators.

\begin{prop}
\label{lem1}
Assume that ${\bf (H1)}$, ${\bf (A1)}$ and ${\bf (A2)}$ hold and for any $(\theta,\theta')\in\Theta^2$ define by $C$ the shortest path $\Pfrak(\theta,\theta')$ length. Then, there exists a sequence of functions $u_n:\nset\to\rset^+$ and a function $v:\rset^+\to\rset^+$ satisfying
\begin{equation}
\label{eq:lemma1_0}
u_n(C)=\frac{\sqrt{C}}{\sqrt{n}}+o(n^{-1/2})\,,\qquad v(\eps)=2\sqrt{d\psi_1\eps}+o(\eps^{1/2})\,,
\end{equation}
such that the pre-computing Metropolis acceptance ratio $\bar{a}(\theta,\theta',\Sfrak)$ (Eq. \ref{eq:nmh}) satisfies
\begin{equation}
\label{eqn:lemma1_bound}
\ex\left|\bar{a}(\theta,\theta',\Sfrak)-a(\theta,\theta')\right|\leq 2c_p^2c_h^2 K_1K_2^{C+2d-1}(\eps)\left\{u_n(C)+v(\eps)\right\}\,.
\end{equation}
In Eq. \eqref{eq:lemma1_0}, $\psi_1<\infty$ is a constant, $n$ is the number of pre-computed GRF realizations for each grid point and $\eps$ is the distance between grid points. In Eq. \eqref{eqn:lemma1_bound}, $K_1$ and $K_2(\epsi)$ are finite constants such that $K_2(\epsi)\to 1$ when $\epsi\downarrow 0$.
\end{prop}

\begin{cor}
\label{cor2}
Define by $\bar{\pi}_i$ the distribution of the $i$-th iteration of the pre-computing Metropolis algorithm implemented with the Full Path estimator. Under Assumptions ${\bf (H1)}$, ${\bf (A1)}$ and ${\bf (A2)}$, we have
\begin{equation}
\label{eq:cor3}
\lim_{i\to\infty}\|\pi-\bar{\pi}_i\|\leq \bar{\kappa}K_2^{2d-1}(\eps)\sum_{c=1}^M K_2^c(\eps)\left\{u_n(c)+v(\eps)\right\}p_c\,,
\end{equation}
where $p_c=\proba\{C=c\}$ is the probability distribution of the path length and
$$
\bar{\kappa}=\left(\lambda-\dfrac{B\rho^\lambda}{1-\rho}\right) 2c_p^2c_h^2 K_1\,.
$$
In Eq. \eqref{eq:cor3}, $u_n$ and $v$ are defined in Eq. \eqref{eq:lemma1_0}.
\end{cor}

Corollary \ref{cor2} states that the asymptotic distance between the pre-computing Markov chain distribution and $\pi$ admits an upper bound that has two main components:
\begin{itemize}
\item $u_n(c)\sim \sqrt{c/n}$ which is related to the variance of each estimator of a normalizing constant ratio estimator,
\item $v(\eps)\sim 2\sqrt{d\psi_1\eps}$ that arises from using a fixed step size grid.
\end{itemize}
This provides useful guidance as to how to tune the pre-computing parameters $n$ and $\eps$. In particular, $n$ should increase with the proposal kernel $h$ variance and $\eps$ should decrease with the dimension of $\Theta$, that is $d$. When $\eps\to 0$ the upper bound of $\|\pi-\bar{\pi}_i\|$ is in $1/\sqrt{n}$ which is in line with the noisy Metropolis rate of \citep{alquier16}. Interestingly, when $\eps\ll 1$, we believe that our bound is tighter thanks to the lower variability of the Full Path estimator compared to the unbiased estimator $\varrho_n$ (Eq. \ref{eq:estimatorZ}) used in the noisy Metropolis algorithm. Indeed, their bound is in $o(K_1^4/\sqrt{n})$ which, given the crude definition of $K_1$, is much looser compared to our $o(K_1/\sqrt{n})$ bound.

The following Proposition shows that when the number of data $n$ simulated at the pre-computing step tends to infinity then $\ex\left|\bar{a}(\theta,\theta',\Sfrak)-a(\theta,\theta')\right|$ vanishes. This result is somewhat reassuring as it suggests that the pre-computing algorithm will converge to the true distribution, asymptotically in $n$, regardless of the grid specification. However, it is not possible to embed this result in the framework developed in \cite{alquier16} as the convergence comes without a rate.

\begin{prop}
\label{prop3}
For any pre-computing Metropolis acceptance ratio that use an estimator of the normalizing constants ratio of the form specified at Eq. \eqref{eq1}:
\begin{multline*}
\ex\left|\bar{a}(\theta,\theta',\Ufrak)-a(\theta,\theta')\right|\leq\int\left|{\psi}{\phi}-\alpha\right|f_n(\rmd \psi\,|\,\phi)(g_n(\phi)-g(\phi))\rmd\phi\\
+\frac{1}{\ex(\Phi_1)\sqrt{n}} \left\{\sqrt{\var(\Psi_1)}+\sqrt{\var(\Phi_1)}\frac{\ex(\Psi_n\left|\zeta\right|)}{\ex(\Phi_1)}\right\}\,,
\end{multline*}
where $\zeta\sim\norm(0,1)$, $f_n$, $g_n$ and $g$ are probability density functions such that $g_n$ converges weakly to $g$.
\end{prop}

\subsection{Toy Example}
We consider in this section the toy example used to illustrate the Exchange algorithm in \cite[Section 5]{murray06}. More precisely, the experiment consists of sampling from the posterior distribution of the
precision parameter $\theta$ arising from the following model:
$$
f(\,\cdot\,|\,\theta)=\norm(0,1/\theta)\,,\qquad p=\text{Gamma}(1,1)\,,
$$
using one observation $y=2$ and pretending that the normalizing constant of the likelihood, namely $Z(\theta)=\int\exp(-\theta y^2/2)\rmd y=\sqrt{2\pi/\theta}$ is intractable. The grid is set as $\Gfrak=\{\tdot_m=m\eps,\,0<m\leq \lfloor 10/m\rfloor\}$. Our objective is to quantify the bias in distribution generated by the pre-computing algorithms.

We consider the situation where the interval between the grid points is $\eps=0.1$ and $n=10$ data are simulated per grid points. Table \ref{tab:var} reports the bias and the variance of the three estimators, \ie the One Pivot, Direct Path and Full Path, of the ratio $Z(\theta)/Z(\theta')$ for three couples $(\theta,\theta')$. This shows that the Full Path estimators enjoys a greater stability than the two other estimators, even when $n$ is relatively small. This is completely in line with the results developed in Propositions \ref{prop1} and \ref{prop2}.

Figure \ref{ex:toy} illustrates the convergence of the three pre-computing Markov chains by reporting the estimated total variation distance between $\bar{\pi}_i$ and $\pi$. We also report the convergence of the exchange Markov chain: this serves as a ground truth since $\pi$ is the stationary distribution of this algorithm. For each algorithm, the total variation distance was estimated by simulating $100,000$ iid copies of the Markov chain of interest and calculating at each iteration the occupation measure. This measure is then compared to $\pi$ which is, in this example, fully tractable. In view of Table \ref{tab:var}, the chains implemented with the One Path and Direct Path estimators converge, as expected, further away from $\pi$ than the Full Path chain.

Interestingly, it can be noted that the Full Path pre-computing chain converges faster than the exchange algorithm. This is an illustration of the observation stated in the introduction regarding the theoretical efficiency of the exchange, compared to that of the plain MH algorithm. Indeed, the pre-computing algorithms aim at approximating MH, and not the exchange algorithm, and should as such inherits MH's fast rate of convergence, provided that the variance of the estimator is controlled.

\begin{table}[h]
\centering
\caption{Bias and variance of the different estimators of the ratio $Z(\theta)/Z(\theta')$ for various couples $(\theta,\theta')$ in the setup of Figure \ref{ex:toy}. The bias and variance were estimated by simulating 10,000 independent realisations of each estimators for each couple $(\theta,\theta')$.\label{tab:var}}
\begin{tabular}{c|cc|cc|cc}
 &  \multicolumn{2}{c|}{$(\theta,\theta')=(1.01,2.06)$} & \multicolumn{2}{|c|}{$(\theta,\theta')=(3.02,0.55)$}& \multicolumn{2}{|c}{$(\theta,\theta')=(0.12,0.94)$}\\
 \hline
{}   & bias   & var.    & bias   & var.   & bias   & var.\\
\hline
FP   & .0007  & .005   & .0004  & .001& .01  & 1.42\\
DP   & .003 & .208   & .003  & .013& .27  & 99.02\\
OP   & .004  & .199   & .003  & .014& .32  & 129.81\\
\end{tabular}
\end{table}

\begin{figure}
\centering
\includegraphics[width=0.9\textwidth]{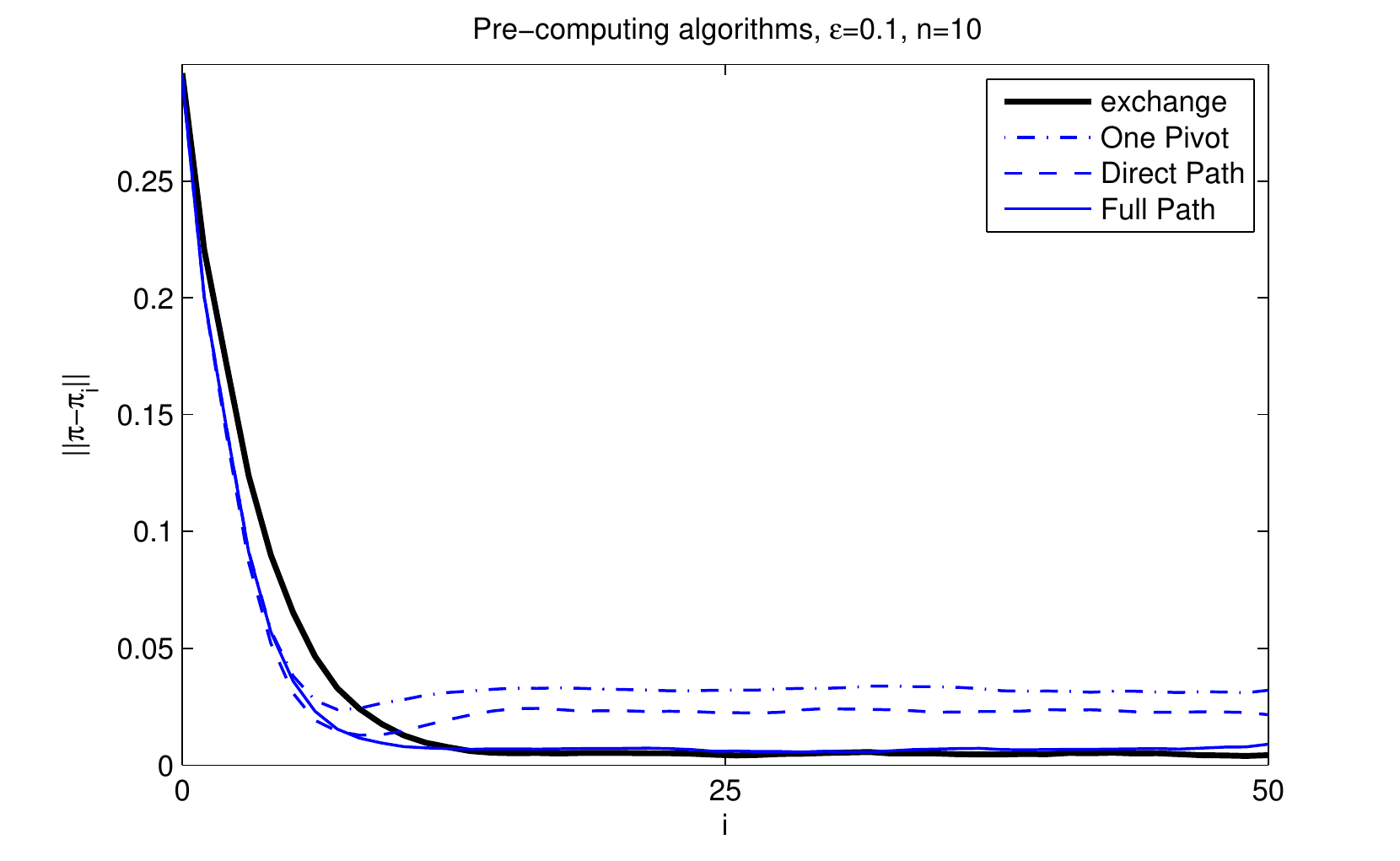}
\caption{Convergence of the pre-computing Metropolis algorithms distribution. Results were obtained from $100,000$ iid copies of the Markov chains initiated with $\mu=p$. All the chains were implemented with the same proposal kernel, namely $\theta'=\theta\exp \sigma\zeta$, $\zeta\sim\norm(0,1)$ and run for 50 iterations. The pre-computing parameters were set to $\eps=0.1$ and $n=10$. Comparing the convergence of the pre-computing chains to that of the exchange (which theoretically converges to $\pi$), we see that the Full Path estimator has a negligible bias. This is not the case for the One Pivot and Direct Path implementations.\label{ex:toy}}
\end{figure}

\section{Results}\label{sec:results}

This section illustrates our algorithm. A simulation study using the Ising model demonstrates the application to a `large' dataset for a single parameter model. More challenging examples are provided with application to a multi-parameter autologistic and Exponential Random Graph Model (ERGM). In the single parameter example we use the estimates of the normalizing constant from Equations \eqref{eq:dpath} and \eqref{eq:fpath}, denoted Full Path and Direct Path respectively. For the single parameter example we compare the pre-computing Metropolis algorithm with the standard exchange algorithm \citep{murray06} and also with a version of the methods in \citet{moores14}. Rather than the Sequential Monte Carlo ABC used in \citet{moores14}, we implemented their pre-computation approach with a MCMC-ABC algorithm \citep{majoram03}. This allowed a fair comparison of expected total variation distance and effective sample size.

\subsection*{MCMC-ABC}

\cite{moores14} used a pre-computing step with Sequential Monte Carlo ABC (see \eg \cite{del06}) to explore the posterior distribution. However, Sequential Monte Carlo has a stopping criterion which results in a finite sample size of values from the posterior distribution. To establish a fair comparison between algorithms whose sample size consistently increases over time, we implemented a modified version of the method proposed in \cite{moores14} using the MCMC-ABC algorithm. The modification made to the MCMC-ABC algorithm amounts to replace a draw $y'\sim f(\cdot|\theta)$ by a distribution that uses the pre-computed data. More precisely, sufficient statistics of a graph at a particular value $\theta$ are sampled from a normal distribution
$$
s\sim \mathcal{N}\left(\mu(\theta,\Ufrak),\sigma^2(\theta,\Ufrak)\}\right)\,.
$$
The parameters $\mu(\,\cdot\,,\Ufrak)$ and $\sigma^2(\,\cdot\,,\Ufrak)$ are interpolated using the mean and variance of the pre-computed sufficient statistics obtained at the grid points. This pre-computing version of ABC-MCMC is described in Algorithm \ref{ABCmcmcalg}.

\begin{algorithm}[H]
\caption{Pre-computing MCMC-ABC sampler}
\label{ABCmcmcalg}
\begin{algorithmic}[1]
\Require Initial distribution $\nu$, a proposal kernel $h$ and ABC tolerance parameter $\epsilon>0$
\State Apply the pre-computing step detailed in \cite{moores14}

$\rightsquigarrow$  pre-computed data $\Ufrak'$.
\State Draw $\theta_0\sim\nu$
\For {$i=1,2,\ldots$}
		\State Draw $\theta'\sim h(\,\cdot\,|\theta_{i-1})$
        \State Calculate the mean $\mu'$ and variance $\sigma'^2$ using the interpolation method in \cite{moores14} and the pre-computed data $\Ufrak'$ for the parameter $\theta'$
		\State Simulate the sufficient statistic $s'\sim \norm\left(\mu,\sigma'^2\right)$
        \State Set $\theta_i=\theta'$ with probability
        $$
        \alpha_{\text{ABC}}(\theta,\theta',\Ufrak):=1\wedge
        \frac{\pi(\theta')h(\theta_{i-1}|\theta')}{\pi(\theta_{i-1})h(\theta'|\theta_{i-1})}\times \1_{|s'-s(y)|<\epsilon}(s')
        $$
        and else set $\theta_{i+1}=\theta_{i}$\,.
	\EndFor

\vspace{.1cm}
\hspace{-1.5cm}\textbf{Return:} The Markov chain $\{\theta_1,\theta_2,\ldots\}$.

\end{algorithmic}
\end{algorithm}

In the multi-parameter example we only compare results of the pre-computing Metropolis with the standard exchange algorithm since the method of \citet{moores14} cannot be implemented in higher dimensions.

\subsection{Ising simulation study}
The Ising model is defined on a rectangular lattice or grid. It is used to model the spatial distribution of binary variables, taking values -1 and 1. The joint density of the Ising model can be written as
$$
f(y|\theta)=\dfrac{1}{Z(\theta)}\exp\left(\theta\sum_{j=1}^M\sum_{i\sim j} y_iy_j\right)\,,
$$
where $i\sim j$ denotes that $i$ and $j$ are neighbours and $Z(\theta)=\sum_y\exp\left(\sum_{j=1}^M\sum_{i\sim j} y_iy_j\theta\right)$.\\
The normalizing constant is rarely available analytically since this relies on taking summation over all different possible realisations of the lattice. For a lattice with $M$ nodes this equates to $2^{\frac{M(M-1)}{2}}$ different possible lattice formations.\\

In this study, 24 lattices of size $80\times80$ were simulated. The true distribution of the graphs were estimated using a long run (24 hours) of the exchange algorithm. Each of the algorithms was run for just over 60 minutes. The pre-computation step of choosing the parameter grid and estimating the ratios for every pair of grid values took approximately 13 minutes. For each of the algorithms we estimated the total variation distance using numerical integration across the kernel density estimates. The values obtained give an indication of which of the chain best matches the long run of the exchange algorithm. The graph in Figure~\ref{tvIsing} is the average of the total variation for each algorithm over all 24 lattices.

\begin{figure}[H]
\centering
\includegraphics[width=0.4\textwidth]{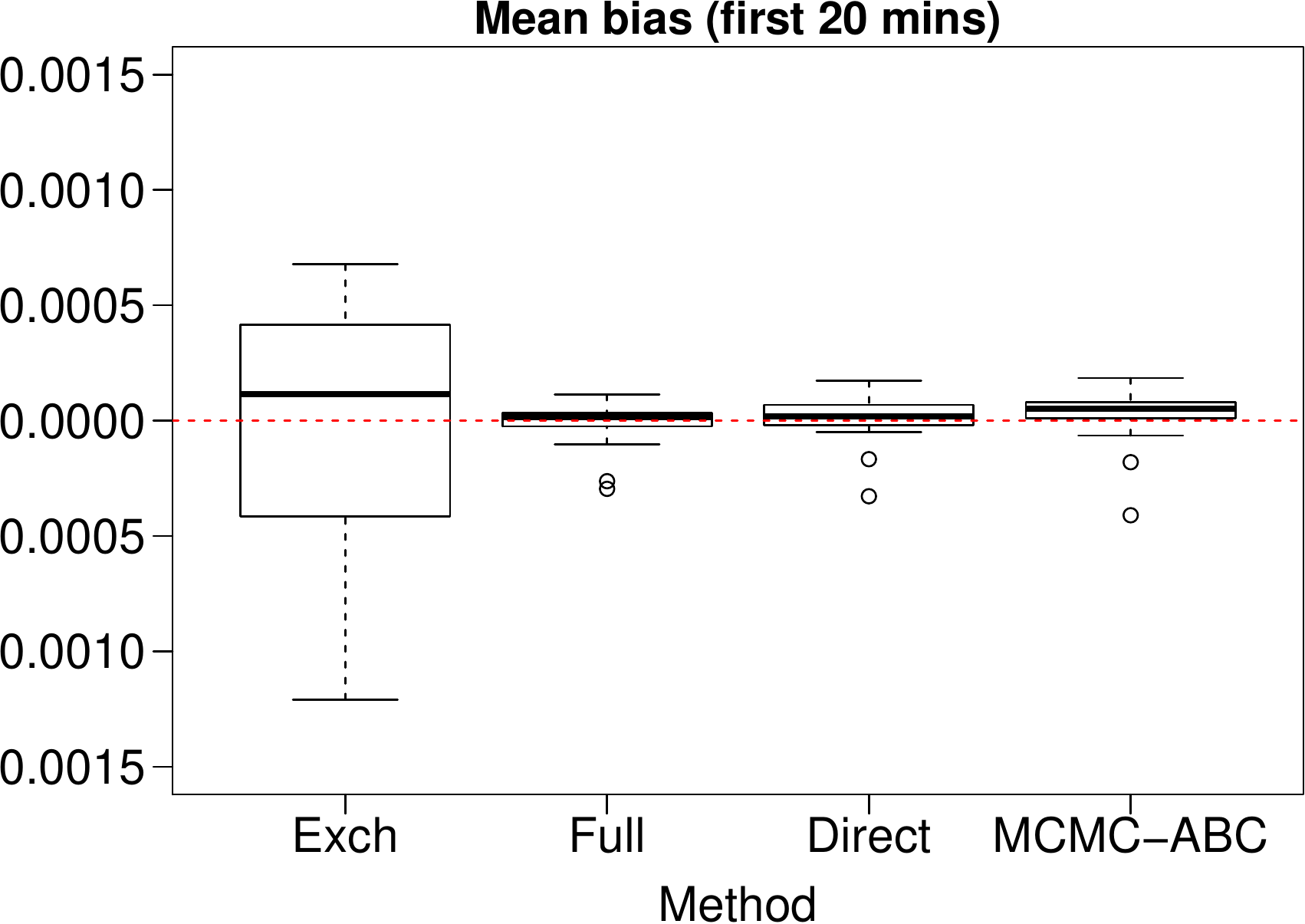}
\includegraphics[width=0.4\textwidth]{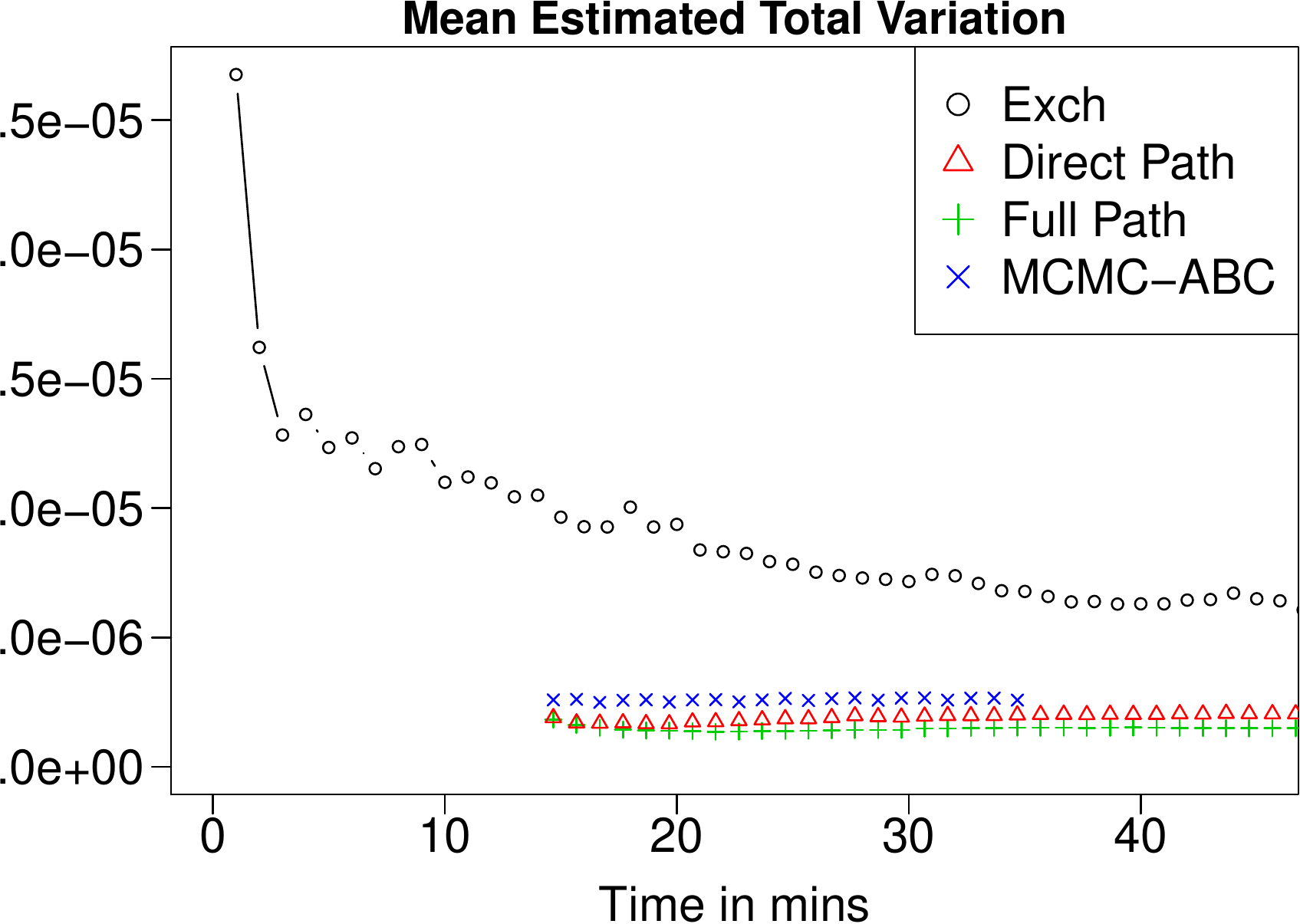}
\includegraphics[width=0.4\textwidth]{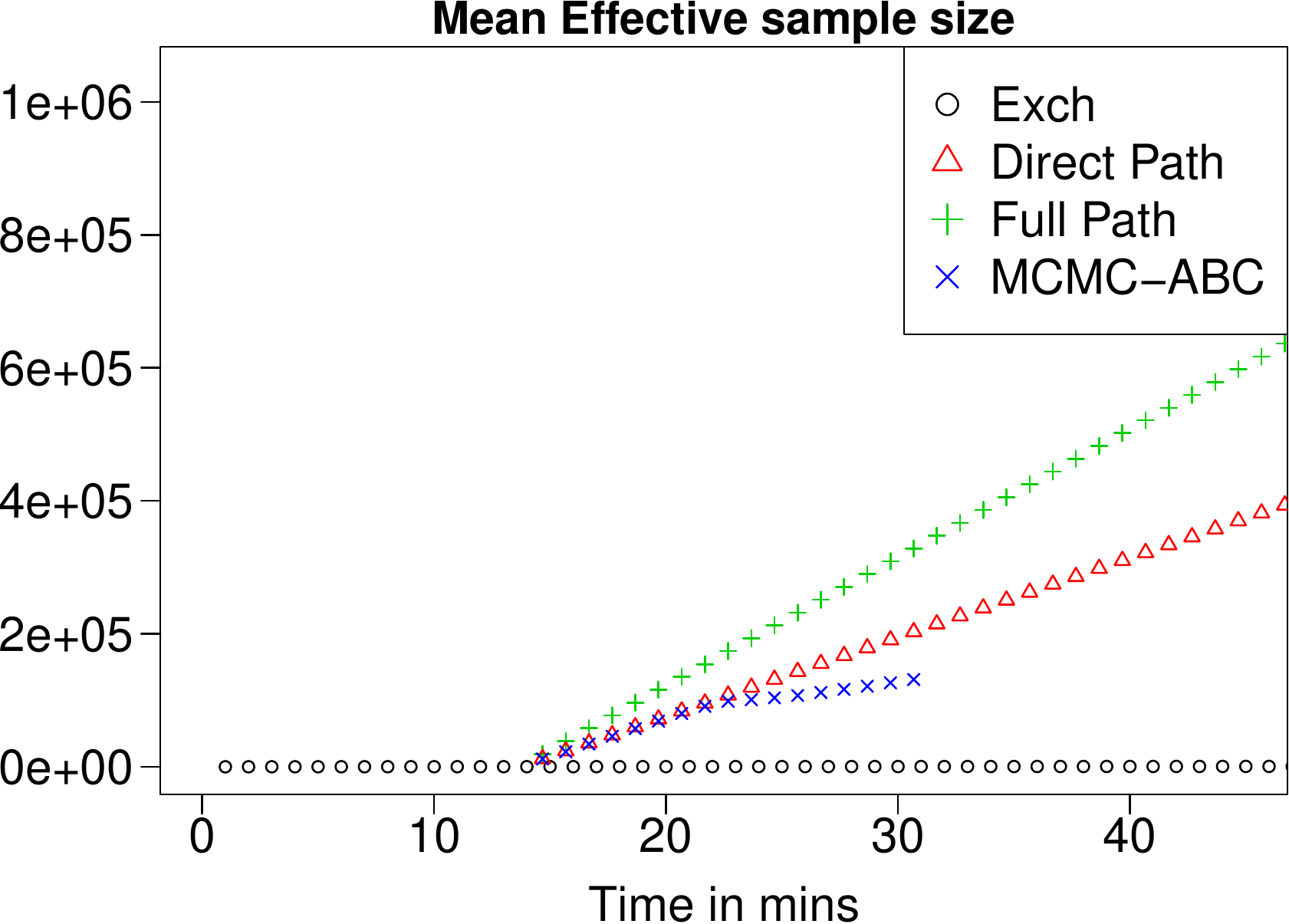}
\caption{Results for the Ising study. The boxplots on the top left show the mean bias of the 24 graphs after the first 20 minutes of computation time: the pre-computing Metropolis algorithm performs the best. The plot on the top right shows the mean estimated total variation of the 24 graphs over time, the pre-computing Metropolis and the MCMC-ABC algorithm both outperform the standard exchange algorithm. The bottom plot shows the effective sample size over time, the pre-computing Metropolis algorithm, implemented with the Full Path estimator performs the best followed by the Direct Path estimator.}
\label{tvIsing}
\end{figure}

The results shown in Figure~\ref{tvIsing} illustrate how the pre-computing Metropolis algorithms (full path and direct path) outperforms the exchange algorithm over time. As more iterations can be calculated per second, the pre-computing
Metropolis algorithm converges quicker. In this simulation, for fairness of comparison, the pre-computing data $\Ufrak$ were re-simulated for each individual graph. Indeed, since all the graphs are on the same
state space, only one single pre-computation step for a large set of parameter values over the full state space could have been sufficient. When analysis is required for many graphs which lie on the same state
space, we only need to carry out the pre-computation step once. We stress that in practice, this situation is common and the speed-up factor obtained by using the pre-computing algorithm would be even more striking.

\subsection{Autologistic Study}
For the second illustration, we extend the Ising model to the autologistic model. The autologistic model is a GRF model for spatial binary data. The likelihood of the autologistic model is given by,
\begin{align*}
f(y|\theta)\propto& \exp(\theta^T s(y))\\
=& \exp(\theta_1 s_1(y) +\theta_2 s_2(y)),
\end{align*}
where $s_1(y)=\sum_{i=1}^N y_i$ and $s_2(y)=\sum_{i\sim j}y_i y_j$ with $i\sim j$ denoting node $i$ and node $j$ are neighbours. $\theta_1$ controls the relative abundance of $-1 $ and $+1$ values while $\theta_2$ controls the level of spatial aggregation. We implement the autologistic model using red deer census data, presence or absence of deer by 1km square in the Grampian region of Scotland \citep{deer96}. Figure~\ref{fig:deer_data} shows the observed data, a red square indicates the presence of deer, while a black square indicates the absence of deer.

\begin{figure}
\centering
\begin{tikzpicture}[x=1pt,y=1pt,scale=0.5]
\input{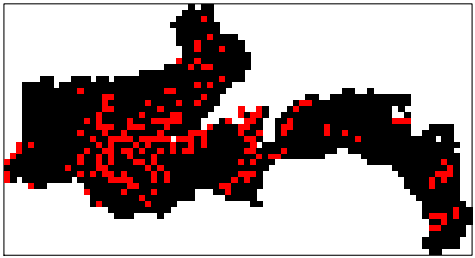}
\end{tikzpicture}
\caption{Presence (red) and absence (black) of red deer in the Grampian region of Scotland.}\label{fig:deer_data}
\end{figure}

A long run (4 hours) of the exchange algorithm was used to set a 'ground truth'. The pre-computing grid points (top left of Figure~\ref{fig:deertv}) were chosen using the method described in Algorithm \ref{alg:grid}.
A total of 124 parameter values were chosen as the values to pre-sampled from. It took just over 45 seconds to choose the grid and calculate the ratios for all pairs of parameter values. The pre-computing Metropolis
algorithms all outperform the exchange algorithm as they converge much quicker, as shown at the top right panel of Figure~\ref{fig:deertv}. In this example, the two different choices of paths yield very similar
results in terms of estimate total variation distance. The pre-computing Metropolis algorithms result in a more accurate mean and variance parameter estimates when compared to the exchange algorithm run for the
same amount of time ; see Table~\ref{deer_table}. When the chains are run for longer, it takes the exchange algorithm 34 minutes to reach the same estimated total variation distance that the pre-computing
Metropolis algorithms takes to reach in $200$ seconds. This illustrates the substantial time saving resulting from the pre-computing Metropolis algorithms.

\begin{table}[ht]
\centering
\caption{Posterior means and variances for the deer data. The table shows that the mean and variance estimates of the noisy exchange are closer to the 'ground truth' long exchange run.}\label{deer_table}
\begin{tabular}{| r | rrrr |}
  \hline
  & $\theta_1$ & & $\theta_2$ &  \\
 & Mean & Variance & Mean & Variance \\
  \hline
Exchange (long) & -0.1435429 	&  0.00028611	&  0.1516334 	& 0.00016096  \\
Exchange 	& -0.1424322 	&  0.00026794	& 0.1530567	& 0.00014771 \\
Full Path 		& -0.1434566 	&  0.00026373	& 0.1519860	& 0.00015384 \\
Direct Path 	& -0.1436186 	&  0.00028256	& 0.1515273	& 0.00016495\\
   \hline
\end{tabular}
\end{table}

\begin{figure}[H]
\centering
\includegraphics[width=0.3928\textwidth]{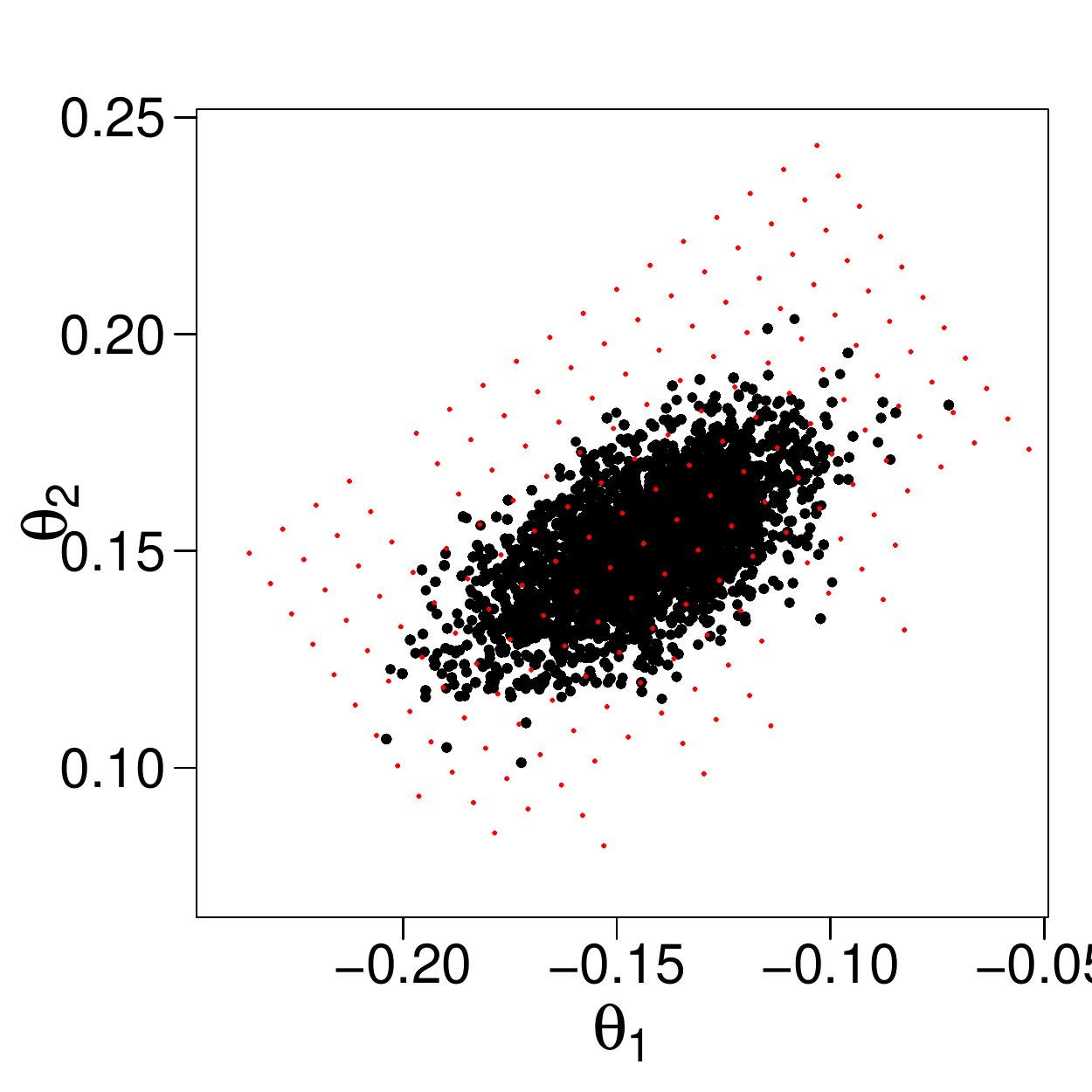} 
\includegraphics[width=0.55\textwidth]{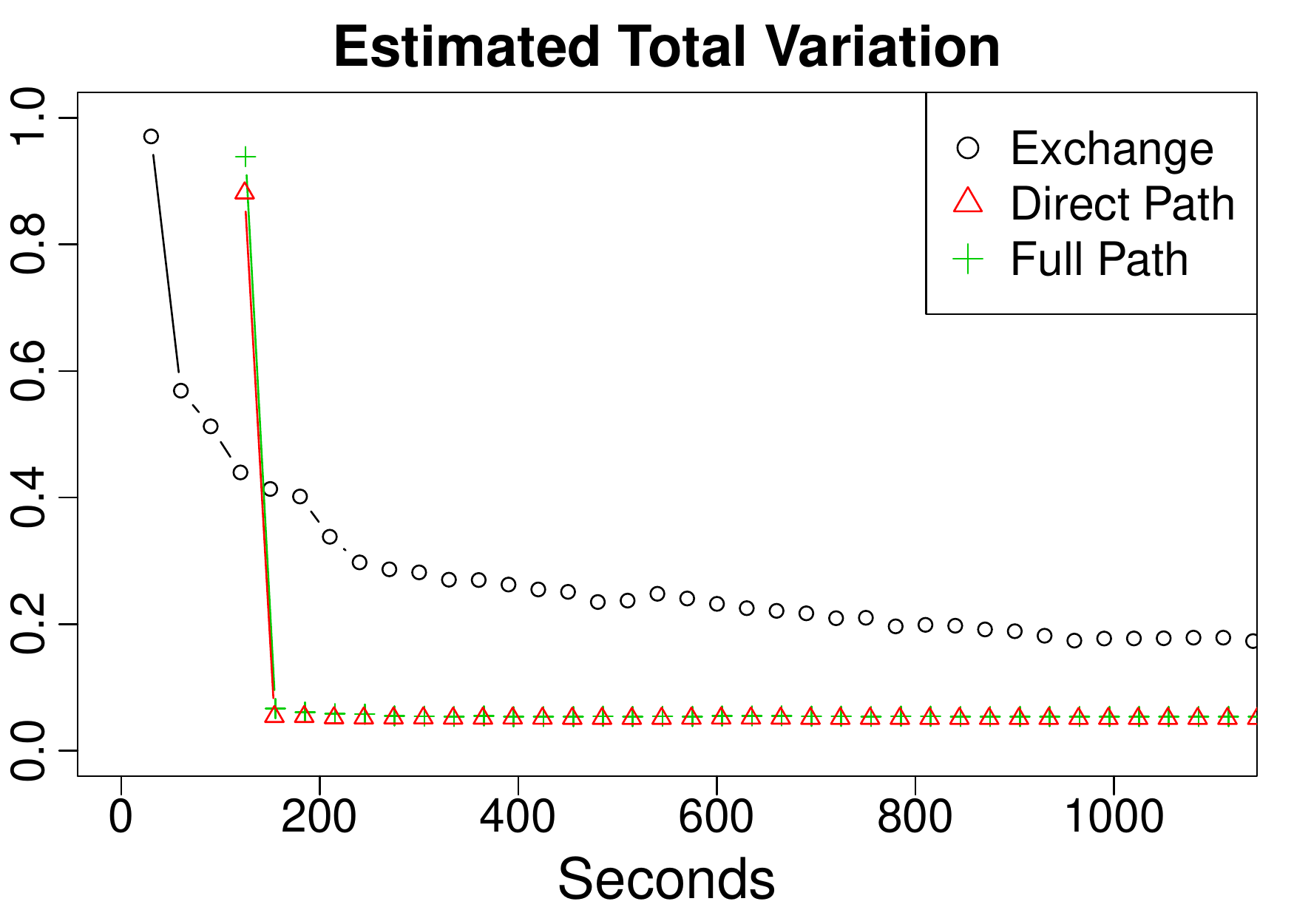}
\includegraphics[width=0.55\textwidth]{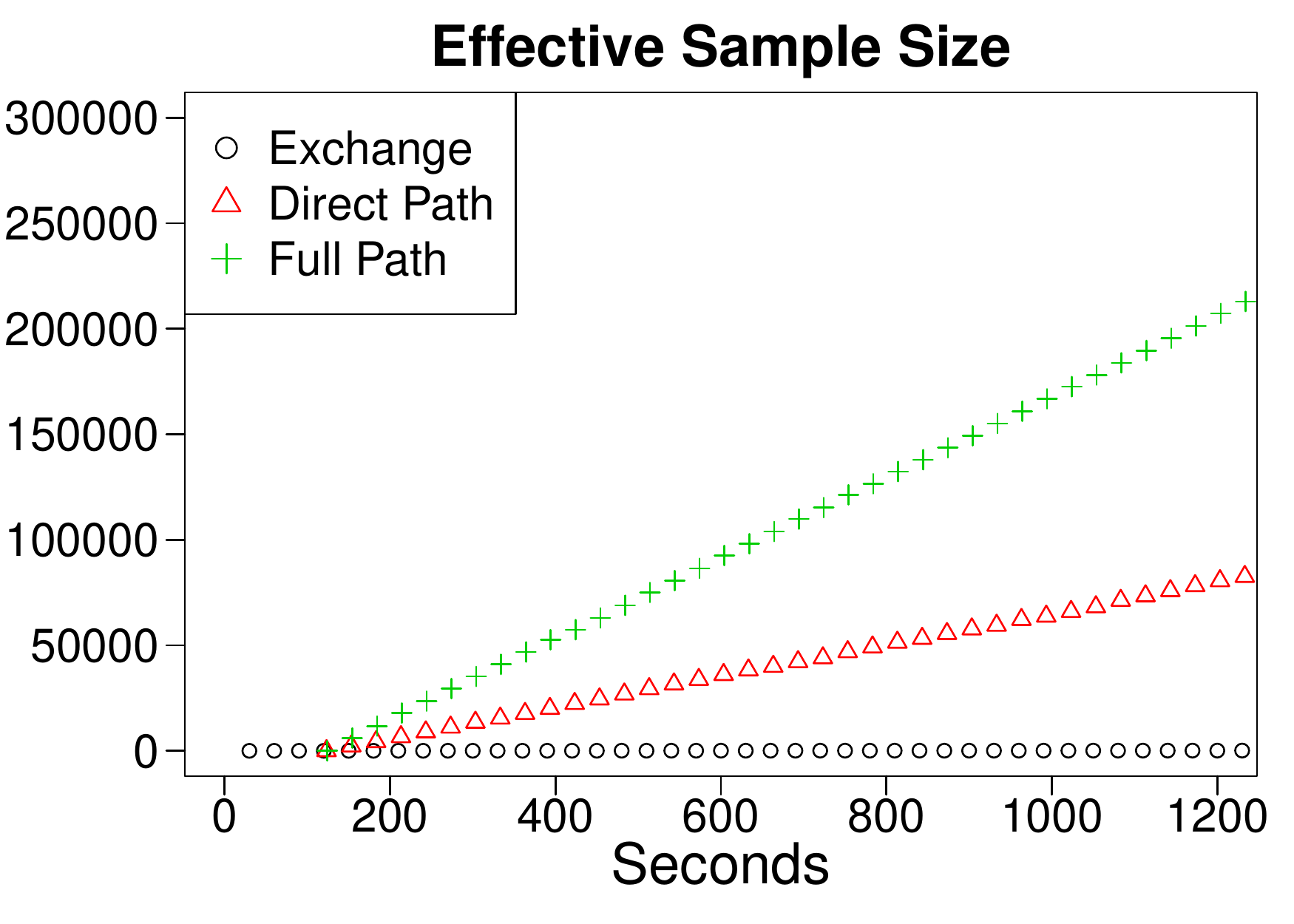}
\caption{Grid for pre-computing (top left) and estimated total variation over time (right). The plot on the top right shows that when the estimated total variation distance between the algorithms and the long exchange is compared, the pre-computing Metropolis algorithms outperform the exchange algorithm. The two versions of the pre-computing Metropolis algorithm also outperform the exchange in terms of effective sample size.}\label{fig:deertv}
\end{figure}

\subsection{ERGM study}

We now show how our algorithms may be applied to the Exponential Random Graph model (ERGM) \citep{robins07}, a model which is widely used in social network analysis. An ERGM is defined on a random adjacency matrix $\Yset$ of a graph on $p$ nodes (or actors) and a set of edges (dyadic relationships) $\{ Y_{ij}: i=1,\dots,M; j=1,\dots,M\}$ where $Y_{ij}=1$ if the pair $(i,j)$ is connected by an edge, and $Y_{ij}=0$ otherwise. An edge connecting a node to itself is not permitted so that $Y_{ii}=0$. The dyadic variables may be undirected, whereby $Y_{ij}=Y_{ji}$ for each pair $(i,j)$, or directed, whereby a directed edge from node $i$ to node $j$ is not necessarily reciprocated.

The likelihood of an observed network $y\in\Yset$ is modelled in terms of a collection of sufficient statistics $\{s_1(y),\dots,s_d(y)\}$, each with corresponding parameter vector
$\theta=\{\theta_1,\dots,\theta_d\}$,
$$
 f(y\,|\,\theta) = \dfrac{q_{\theta}(y)}{Z(\theta)} =
 	\frac{\exp\left\{ \sum_{l=1}^m \theta_l s_l(y) \right\}}{Z(\theta)}\,.
$$
Typical statistics include the observed number of edges and  the observed number of two-stars, which is the number of configurations of pairs of edges which share a common node. Those statistics are usually defined as
$$
s_1(y):=\sum_{i<j}y_{ij}\,,\qquad s_2(y):=\sum_{i<j<k}y_{ik}y_{jk}\,.
$$
It is also possible to consider statistics which count the number of triangle configurations, that is, the number of configurations in which nodes $\{i,j,k\}$ are all connected to each other.

\subsubsection{Karate dataset}
We consider Zachary's karate club \citep{zachary77} which represents the undirected social network graph of friendships between 34 members of a karate club at a US university in the 1970s.

\begin{figure}[H]
\centering
\includegraphics[width=0.5\textwidth]{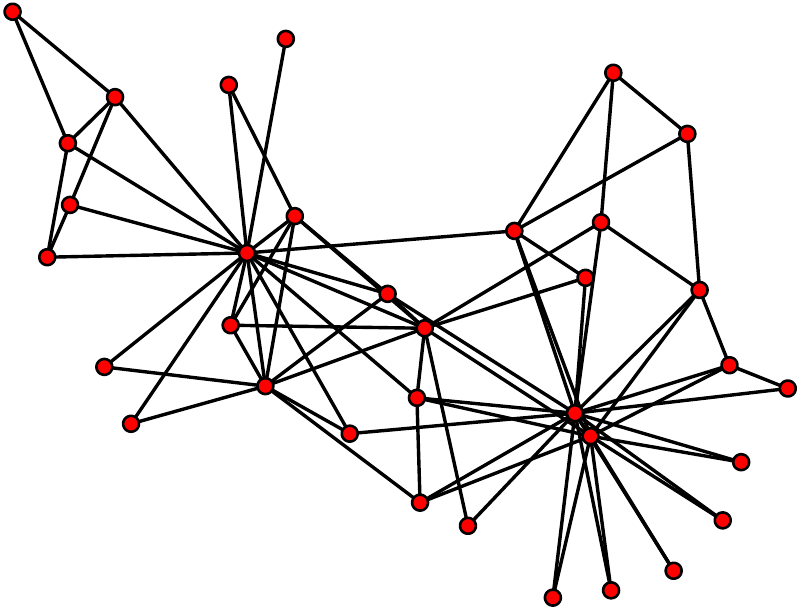}
\caption{Karate club data.}
\label{karatedata}
\end{figure}
We consider the following two-dimensional model,
$$
f(y\,|\,\theta)=\frac{1}{Z(\theta)}\exp\left\{\theta_1s_1(y)+\theta_2s_2(y)\right\},
$$
where $s_1(y)$ is the number of edges in the graph and $s_2(y)$ is the number of triangles in the graph.

A long run of the exchange algorithm was again used to set a `ground truth'. The pre-computing step took roughly 30 seconds to set the $M=191$ parameter values constituting the grid and to calculate the estimated normalizing ratio between each pair of parameter values using $n=1,000$ simulated graphs. The mean and variance of the parameter estimates for the noisy exchange algorithms using the two different paths and a short run of the exchange algorithm are compared in Table~\ref{karatetable}. Figure~\ref{karatetv} shows the choice of parameter for pre-processing (left) and the estimated total variation distance over time (right). Some grid points lie beyond the posterior distribution high density region, indicating that some graphs sampled from the tail regions could have been avoided. In practice however, it was found that allowing the grid to span beyond the posterior distribution high density regions gave much better results. The two versions of the pre-computing Metropolis algorithm outperform the exchange algorithm in the estimated total variation distance over time.

\begin{table}[H]
\centering
\caption{Posterior means and variances for the karate data.}
\label{karatetable}
\begin{tabular}{| l|c c c c |}
\hline
			&	Edge			&			&	Triangle	&	\\
			&	Mean		&	Var		&	Mean	&	Var\\
\hline
Exchange (long) & -2.0471 &  0.0962 & 0.3807 & 0.0306 \\
Exchange & -2.1758 & 0.0739  & 0.4670 &  0.0254\\
Full Path & -2.3328 & 0.0991  & 0.4922 & 0.0210 \\
Direct Path & -2.1645 & 0.1095  & 0.4518 & 0.0454 \\
\hline
\end{tabular}
\end{table}

\begin{figure}[H]
\centering
\includegraphics[width=0.3928\textwidth]{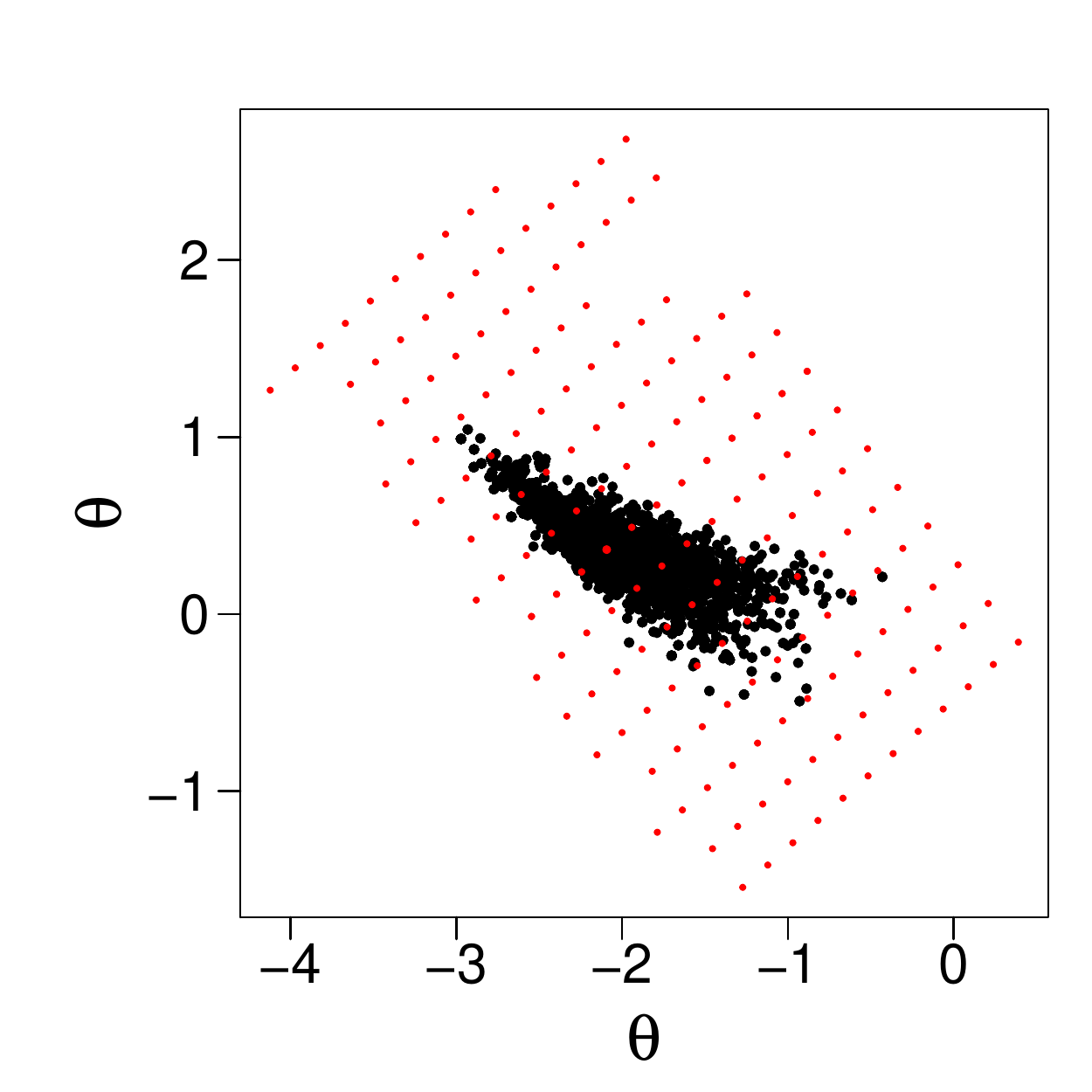} 
\includegraphics[width=0.55\textwidth]{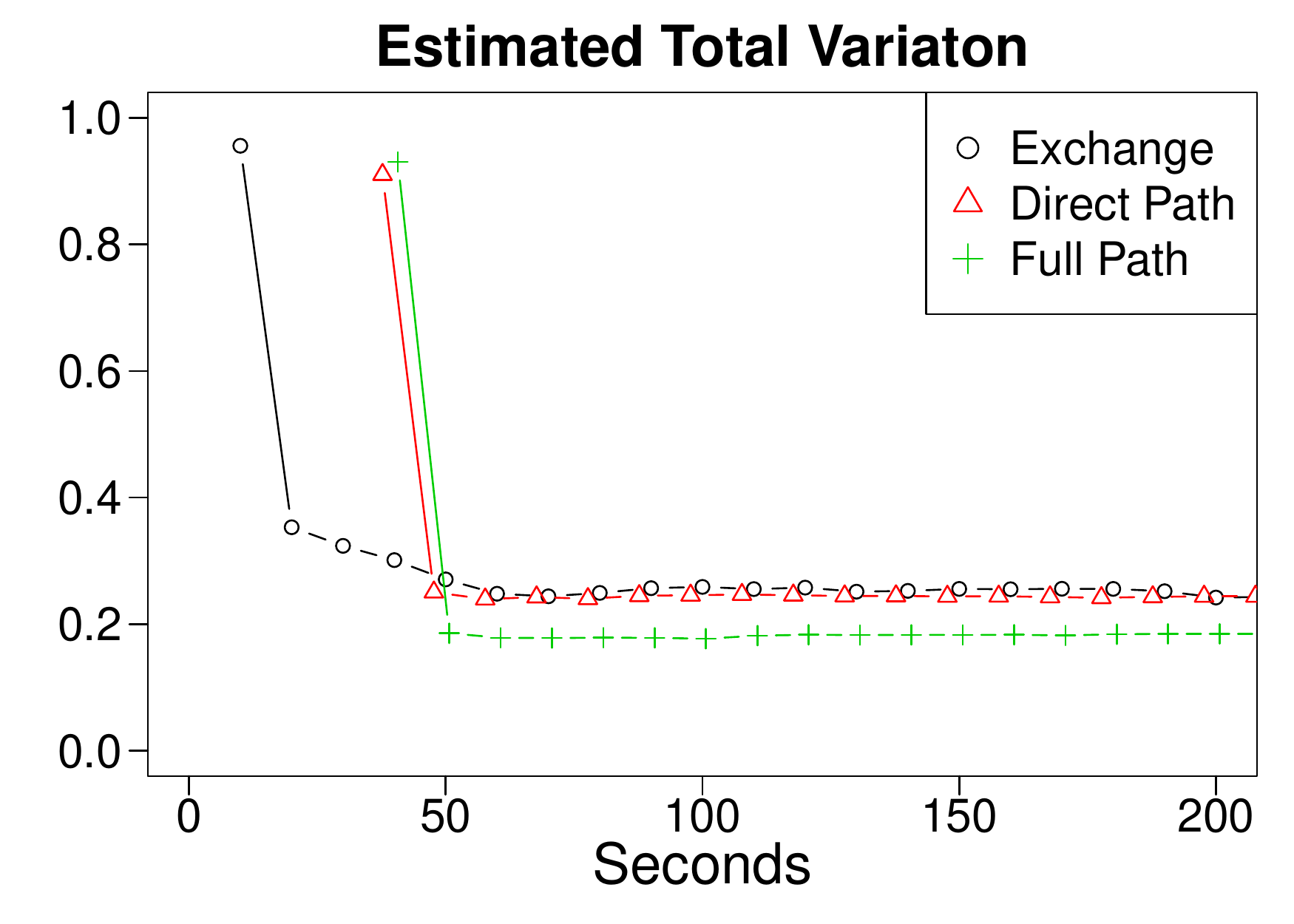}
\includegraphics[width=0.55\textwidth]{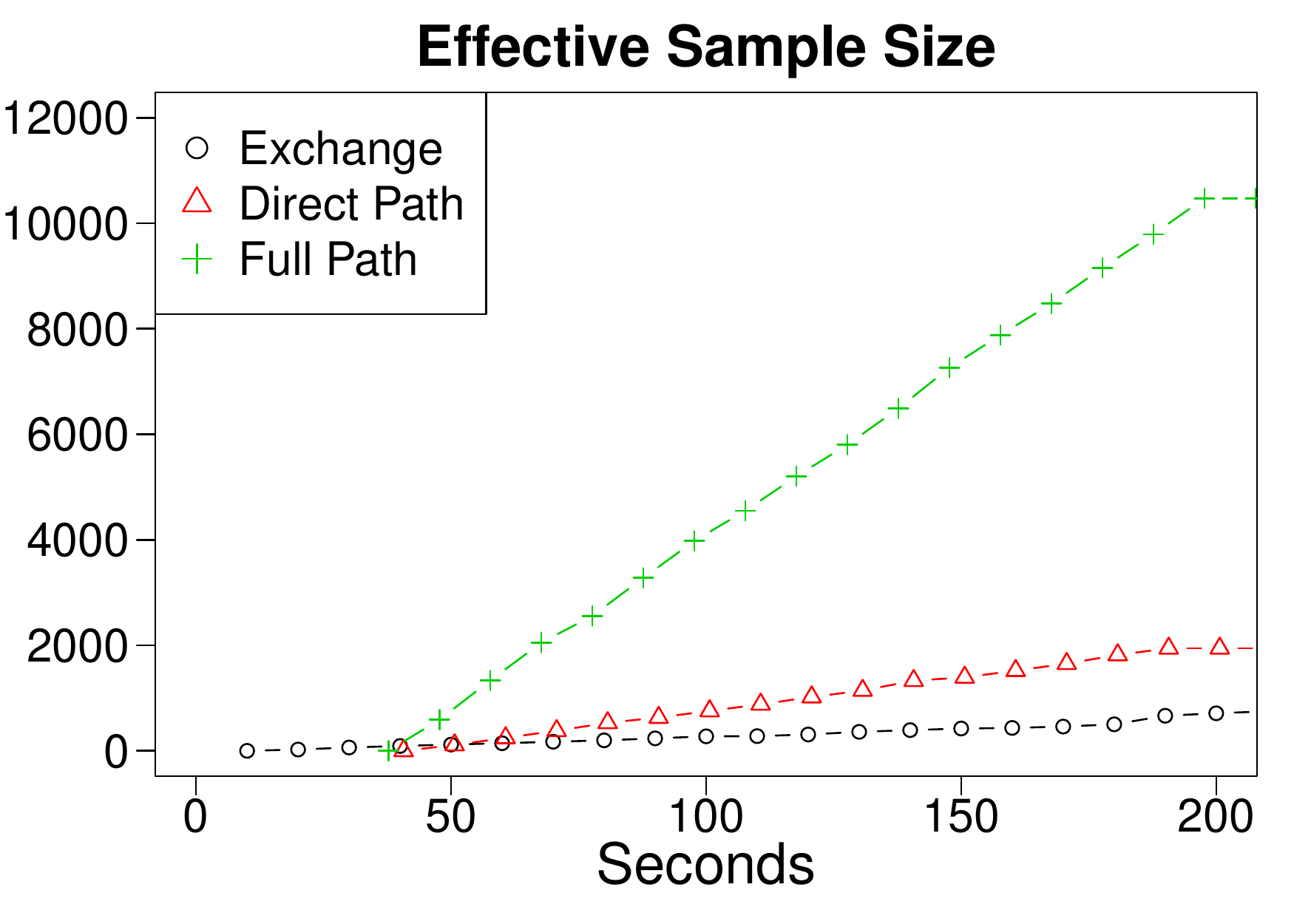}
\caption{Grid for pre-computing (left) and estimated total variation distance over time (right). The pre-computing Metropolis algorithms outperform the exchange in terms of estimated total variation distance. The effective sample size of the pre-computing algorithms is much higher than the exchange.}
\label{karatetv}
\end{figure}

\section{Conclusion}
This paper considers including an offline, easily parallelizable, pre-computing step as a way to overcome the computational bottleneck of certain variants of the Metropolis algorithm. In particular, we show how such a strategy can be efficient when inferring a doubly-intractable distribution, a situation that typically arises in the study of Gibbs random fields. The pre-computing Metropolis algorithms that we develop in this paper somewhat borrow from previous pre-computing algorithms (see \eg \cite{moores15}) but scale better to higher dimensional settings. We however note that our method would be impractical for very high dimensions. Yet, the limit on the number of dimensions is similar to the limit on the INLA method  \citep{rue09}, which has seen widespread use in many areas.

The pre-computing Metropolis algorithms are noisy MCMC algorithms in the sense that the posterior of interest is not the invariant distribution of the Markov chain. However, we establish, under certain conditions, some theoretical results showing that the pre-computing Metropolis distribution converges into a ball centered on the true posterior distribution. Interestingly, the ball radius can be made arbitrarily small according to the pre-computing parameters, namely the space between grid points and the number of auxiliary data simulated per grid points. Our main contribution to the theoretical analysis of approximate Markov chains is twofold:
\begin{itemize}
\item In contrast to estimators of the Metropolis acceptance ratio that have been used in the approximate MCMC literature (see \eg \cite{alquier16}, \cite{medina2016stability} and \cite{bardenet2014towards}), the different estimators considered in this paper (\ie the One Pivot, the Direct Path and the Full Path) are all biased. We stress that, when computational time is not an issue, there is no particular gain in efficiency using biased estimators but biasedness is an inevitable by-product when estimators make use of pre-computed data exclusively.
\item A recurrent outcome from the research on approximate MCMC methods highlights the importance of using estimators of the Metropolis acceptance ratio with \textit{small} variance. We refer for instance to the aforementioned works and \cite{bardenet2015markov}, \cite{quiroz2017speeding} and \cite{stoehr2017noisy}. In the context of estimating a ratio of normalizing constants, we argue that the pre-computing step allows to specify low variance estimators, yet biased, at low computational cost by considering intermediate grid points, an idea that has been exploited by the Full Path estimator.
\end{itemize}

The empirical results show that in time normalized experiments, the pre-computing Metropolis algorithms provide accurate and efficient inference that outperform existing techniques such as the exchange algorithm \citep{murray06}.

Focus for future research will examine alternative methods that would allow inference of higher dimensional models. As it stands, the curse of dimensionality implies an exponential growth of the number of grid points,
which makes our pre-computing step far too computationally intensive to be implemented in this setting. A way to overcome this challenge would be to design the grid adaptively, \ie as the Markov chain is being simulated,
in order to avoid unnecessary simulations at grid points whose vicinity is never visited by the Markov chain. Even though such a strategy is straightforward to implement, the theoretical analysis of the resulting
algorithm is more involved. Indeed, it calls for results on ergodicity of approximate adaptive Markov chain, a research topic which is for now essentially unexplored.

\section*{Acknowledgements}
The Insight Centre for Data Analytics is supported by Science Foundation Ireland under Grant Number SFI/12/RC/2289. Nial Friel's research was also supported
by a Science Foundation Ireland grant: 12/IP/1424.

\bibliographystyle{chicago}
\bibliography{normbib}

\appendix
\section*{Appendix}

\paragraph{Variance of the estimators}

\begin{proof}[Proof of Proposition \ref{prop1}]
Denoting by $v_n$ the variance in Eq. \eqref{eq:prop1}, it comes
$$
n v_n(\theta,\theta')=\ex_{\theta'}\exp2(\theta-\theta')\T s(y)-\left(\frac{Z(\theta)}{Z(\theta')}\right)^2
$$
Showing that $n v_n(\theta,\theta')\in\bigO(\|\theta-\theta'\|^2)$ follows from Taylor expanding the function around $\theta$:
\begin{equation*}
n v_n(\theta,\theta')=1+2h\T \ex_{\theta'}s(y)+\ex_{\theta'} f(h,y)-\frac{1}{Z(\theta')^2}\left(Z(\theta')^2+2Z(\theta')h\T\nabla Z(\theta')+\bigO(\|h\|^2)\right)
\end{equation*}
where we have introduced $h=\theta-\theta'$ and $f=\bigO(\|h\|^2)$. Noting that $f=\bigO(\|h\|^2)\Rightarrow \ex_{\theta'}f(h,y)=\bigO(\|h\|^2)$ and $\nabla Z(\theta')\slash Z(\theta')=\ex_{\theta'}S(y)$ yields:
\begin{equation*}
n v_n(\theta,\theta')=\bigO(\|h\|^2)\,,
\end{equation*}
which concludes the proof.
\end{proof}

\begin{proof}[Proof of Proposition \ref{prop2}]
Note that
\begin{multline*}
v_n^{\DP}-v_n^{\FP}=\ex\left(\frac{R_n^1}{\Psi_n}\right)^2\left\{\left(R_n^{2C}\right)^2-\left(R_n^2\times\cdots\times R_n^C\right)^2\right\}\\
-\left\{\ex\frac{R_n^1}{\Phi_n}R_n^{2C}\right\}^2+
\left\{\ex\frac{R_n^1}{\Phi_n}R_n^{2}R_n^3\times\cdots\times R_n^C\right\}^2\,,
\end{multline*}
which under the assumption that $\Phi_n$ and $\Psi_n$ are independent yields
\begin{multline*}
v_n^{\DP}-v_n^{\FP}=\ex\left(\frac{R_n^1}{\Psi_n}\right)^2\left\{\ex\left(R_n^{2C}\right)^2-\ex\left(R_n^2\times\cdots\times R_n^C\right)^2\right\}\\
-\left(\ex\frac{R_n^1}{\Psi_n}\right)^2\left\{\left(\ex R_n^{2C}\right)^2-\left(\ex R_n^2\times\cdots\times R_n^C\right)^2\right\}\,.
\end{multline*}
Equation \eqref{eq:prop2:1} holds with $\gamma=\ex\left({R_n^1}\slash{\Psi_n}\right)^2$ as a result of $\ex R_n^{2C}=\ex R_n^2\times\cdots\times R_n^C=Z(\tdot_1)/Z(\tdot_C)$.

For simplicity of notation, define $R_n:=R_n^2\times\cdots\times R_n^C$ and $X_n=\log R_n$. For large $n$, $R_n^i$ can be approximate by a truncated normal (in the positive range) $\bR_n^i\sim \norm_+(\mu_i,(1/n)\sigma_i^2)$, where $\mu_i:=Z(\tdot_{i-1})/Z(\tdot_i)$ and $\sigma_i=\var\{\exp(\tdot_{i-1}-\tdot_i)\T s(X_i)\}$. It can be noted that, upon reparameterization of the sufficient statistics vector (in the space spanned by the matrix $V$ column vectors), we have $\sigma_i=\var\{\exp \eps s_i\}$ where $s_i$ is the projection on the only one dimension where $\tdot_{i-1}$ and $\tdot_i$ are not equal of the sufficient statistics $s(X_i)$, $X_i\sim f(\,\cdot\,|\,\tdot_i)$. Applying the delta method yields that $X_i$ can be approximate by
\begin{equation}
\label{eq:prop1:0}
\bX_n^i:=\log \bR_n^i\sim \norm\left(\log\mu_i,\frac{\sigma_i^2}{n\mu_i^2}\right)\,.
\end{equation}
Define $\bX_{n,C}:=\sum_{i=1}^{C-1} \bX_n^{i+1}$ and note that the sequence $\{\bX_n^1,\bX_n^2,\ldots\}$ satisfies a Lyapunov condition \ie
\begin{equation}
\label{eq:prop2:Lyap}
\lim_{C\to\infty}\frac{\sum_{i=1}^{C}\ex\left|\bX_n^i-\ex\bX_n^i\right|^4}{\left\{\sum_{i=1}^{C}\var(\bX_n^i)\right\}^4}=0\,.
\end{equation}
Indeed, it can be checked that the fourth central moment of a Gaussian random variable verifies $\ex\left|\bX_n^i-\ex\bX_n^i\right|^4=3\var(\bX_n^i)^2$. Moreover since the $\sigma_i$'s are bounded, there exists two numbers $0<m<M<\infty$ such that $m\leq \var(\bX_n^i)\leq M$. This allows to justify \eqref{eq:prop2:Lyap} since
$$
\frac{\sum_{i=1}^{C}\var(\bX_n^i)^2}{\left\{\sum_{i=1}^{C}\var(\bX_n^i)\right\}^4}\leq \frac{\sum_{i=1}^{C}\var(\bX_n^i)^2}{m^2C^2\left\{\sum_{i=1}^{C}\var(\bX_n^i)\right\}^2}\leq
\frac{1}{C\left\{\sum_{i=1}^{C}\var(\bX_n^i)\right\}^2}\left(\frac{M}{m}\right)^2\,,
$$
whose right hand side goes to 0 when $C\to \infty$. In virtue of \eqref{eq:prop2:Lyap}, a central limit holds for $\bX_{n,C}$ and in particular, asymptotically in $C$,
\begin{equation}
\label{eq:prop1:01}
\bX_{n,C}\Rightarrow \norm\left(\sum_{i=1}^{C-1}\ex \bX_n^{i+1},\sum_{i=1}^{C-1}\var\bX_n^{i+1}\right)\,,
\end{equation}
which implies that $\bR_n$ is log-normal and, as a consequence,
\begin{equation}
\label{eq:prop1:2}
\var \bR_{n}=\left\{\exp\left(\var\bX_{n,C}\right)-1\right\}\exp\left(2\ex\bX_{n,C}+\var\bX_{n,C}\right)\,.
\end{equation}
First note that combining \eqref{eq:prop1:0} and \eqref{eq:prop1:01}
\begin{equation}
\label{eq:prop1:3}
\ex\bX_{n,C}=\log\frac{Z(\tdot_1)}{Z(\tdot_C)}\,,\qquad\var\bX_{n,C}=\frac{1}{n}\sum_{i=2}^C\left\{\frac{Z(\tdot_i)}{Z(\tdot_{i-1})}\right\}^2\left\{\var\exp(\eps s_i)\right\}^2
\end{equation}
and
$$
\var\exp(\eps s_i)=\var\left(1+\eps \sum_{j=1}^\infty \frac{\eps^{j-1} s_i^j}{j!}\right)=
\eps^2 v_i\,,\qquad v_i:=\var\left(\sum_{j=1}^\infty \frac{\eps^{j-1} s_i^j}{j!}\right)\,.
$$
Putting together with \eqref{eq:prop1:3}, we have:
$$
\var \bX_{n,C}=\frac{\eps^4}{n}\sum_{i=2}^C\left\{v_i\frac{Z(\tdot_i)}{Z(\tdot_{i-1})}\right\}^2\,,
$$
which eventually using \eqref{eq:prop1:2} leads to
\begin{equation}
\var \bR_{n}=\frac{\eps^4}{n}\sum_{i=2}^C\left\{v_i\frac{Z(\tdot_i)Z(\tdot_1)}{Z(\tdot_{i-1})Z(\tdot_C)}\right\}^2+o\left(\eps^4/n\right)\,.
\end{equation}
\end{proof}

\begin{rmk}[On the proof of Proposition \ref{prop2}]

Even though Proposition \ref{prop1} is established under the assumption that $\Psi_n$ and $\Phi_n$ are independent, note that this can be relaxed if the Direct Path estimator includes one more grid point in its path \ie if $\Psi_n^{\DP}$ estimates $Z(\theta)/Z(\tdot_1)\times Z(\tdot_1)/Z(\tdot_{C-1})\times Z(\tdot_{C-1})/Z(\tdot_{C})$. When $\epsilon$ is small, we expect that the Direct Path estimator and this alternate version would be highly similar. The result of comparison between the variances of the Full Path estimator and this alternate version of the Direct Path estimator holds without the independence assumption.
\end{rmk}

\begin{rmk}[On the proof of Proposition \ref{prop2}]
Unlike $n$ and $\eps$, the path length $C$ in the Full Path estimator is a random variable that depends on $(\theta,\theta')$. Therefore, one can critically comment on the use of a central limit theorem in $C$ that is needed to establish Eq. \eqref{eq:prop2:2}. However, we insist that $C$ could be made arbitrarily as large as needed by using a path connecting $\theta$ to $\theta'$ that is long enough. This type of path should, however, not use a same grid point twice in order to satisfy the independence assumption of the central limit theorem.
\end{rmk}

\paragraph{Convergence of the pre-computing transition kernel}
We preface the proof of Proposition \ref{lem1} with the following Lemma.

\begin{lemma}
\label{lem:iid}
Let $\bar{X}_n^1,\ldots,\bar{X}_n^r$ be $r$ \iid sample mean estimators \ie for $j\in\{1,\ldots,r\}$, $\bar{X}_n^j=n^{-1}\sum_{k=1}^n X_{j,k}$, $X_{j,1},\ldots,X_{j,n}\sim_{\iid} \pi_j$, where $\pi_j$ is any distribution. Assume that there exists a positive number $M>0$ such that for all $j$, the support of $\pi_j$ is such that $\mathrm{supp}(\pi_j)\subseteq(0,M)$. Then:
\begin{equation*}
\var(\bar{X}_n^1\times\cdots\times\bar{X}_n^r)\leq M^{2r}\left\{\left(1+\frac{1}{n}\right)^r-1\right\}\,.
\end{equation*}
\end{lemma}

\begin{proof}
This follows from the variance of a product of independent random variables. More precisely, $\var(\bar{X}_n^1\times\cdots\times\bar{X}_n^r)$ is a sum of $2^r-1$ products of $r$ positive factors. Each factor is either a squared expectation $(\ex\bar{X}_n^j)^2$ or a variance $\var \bar{X}_n^j$ so that one of the $2^r-1$ products that contains $k$ variances  ($k>0$) and $r-k$ squared expectations is
\begin{equation*}
p_k:=\var(\bar{X}_1^n)\var(\bar{X}_2^n)\times\cdots\times\var(\bar{X}_k^n)(\ex\bar{X}_{k+1}^n)^2(\ex\bar{X}_{k+2}^n)^2\times\cdots\times(\ex\bar{X}_{r}^n)^2\,.
\end{equation*}
Note that $p_k$ can be reexpressed as
\begin{equation}
\label{eq_lem_1}
p_k=\frac{1}{n^k}\var({X}_{1})\var({X}_{2})\times\cdots\times\var({X}_{k})(\ex{X}_{k+1})^2(\ex{X}_{k+2})^2\times\cdots\times(\ex{X}_{r})^2\,,
\end{equation}
where for simplicity we have defined $X_j\sim\pi_j$ in Eq. \ref{eq_lem_1}. Interestingly, $p_k$ can be uniformly bounded in $k$ as follows:
\begin{equation}
\label{eq:conv:1}
p_k\leq \frac{1}{n^k}\ex({X}_{1}^2)\ex({X}_{2}^2)\times\cdots\times\ex({X}_{k}^2)(\ex{X}_{k+1})^2(\ex{X}_{k+2})^2\times\cdots\times(\ex{X}_{r})^2
\leq \frac{M^{2r}}{n^k}\,.
\end{equation}
Since there are ${r\choose k}$ terms that have $k$ variances and $r-k$ squared expectations, their sum $\bar{p}_k$ can be bounded using the uniform bound provided in Eq. \eqref{eq:conv:1} so that
$$
\bar{p}_k\leq {r\choose k}\frac{M^{2r}}{n^k}\,.
$$
The proof is completed by rearranging the sum of $2^r-1$ products by aggregating those products that have the same number $k$ of variance factors, \ie
$$
\var(\bar{X}_n^1\times\cdots\times\bar{X}_n^r)= \sum_{k=1}^r\bar{p}_k\leq M^{2r}\sum_{k=1}^r{r\choose k}\frac{1}{n^k}=M^{2r}\left(1+\frac{1}{n}\right)^r-1\,.
$$
\end{proof}

\begin{proof}[Proof of Proposition \ref{lem1}]
For notational simplicity, $\ex$ is the expectation operator under the distribution of the pre-computed data $\Ufrak$. Under the assumptions of Lemma \ref{lem1}, the two following constants
\begin{equation}
\label{eq:assu}
T:=\sup_{\theta\in\Theta}\|\theta\|\,,\qquad \psi_1=:\sup_{x\in\Yset}\sup_{j\in\{1,\ldots,d\}}s_j(x)
\end{equation}
are finite. We first state three inequalities that are immediate consequences from the grid geometry:

\begin{itemize}
\item Noting that $|(\theta-\theta')\T s(x)|\leq \psi_1\|\theta-\theta'\|$, for any $(\theta,\theta')\in\Theta^2$ and $x\in\Yset$, we have:
\begin{align}
\exp\left(-2T\psi_1\right)\leq \dfrac{q_{\theta}(x)}{q_{\theta'}(x)}=\exp(\theta-\theta')\T s(x)\leq\exp\left(2T\psi_1\right):=K_1.\nonumber
\end{align}
\item For two neighboring points $\tdot_k$ and $\tdot_m$ in the pre-computed grid, there exists $j\in\{1,\ldots,d\}$ such that
$
|(\tdot_{k}-\tdot_{m})\T s(x)|=\pm \eps s_j(x)\,,
$
which yields
\begin{align}
1/K_2(\eps)\leq \dfrac{q_{\tdot_{k}}(x)}{q_{\tdot_{m}}(x)}=\exp(\tdot_{k}-\tdot_{m})\T s(x)\leq\exp\left(\eps\psi_1\right):=K_2(\eps)\,.\nonumber
\end{align}
\item For any $\theta\in\Theta$, there is a point $\tdot\in\Gfrak$ such that for all $j\in\{1,\ldots,d\}$, $|\theta_j-\tdot_j|<\epsilon$, we have
\begin{equation}
\label{eq:bounds_K2}
  1/K_2^d(\epsi)\leq\dfrac{q_{\theta}(x)}{q_{\tdot_m}(x)}\leq K_2^d(\epsi)\,.
\end{equation}
\end{itemize}
We will intensively use the result from Lemma \ref{lem:iid} on the variance of a product of independent estimators and the fact that for any random variable $X$
\begin{equation}
\label{eq:variance}
\exists\,M\in\rset\quad\text{s.t.}\quad X\leq M\quad\Rightarrow\quad\var X\leq \ex X^2\leq M^2\,.
\end{equation}
We recall that in the pre-computing Metropolis algorithm, the normalizing constant ratio $Z(\theta)/Z(\theta')$ is estimated by
$$
\rho_n(\theta,\theta',\Ufrak)=\frac{\Psi_n(\theta,\theta',\Ufrak)}{\Phi_n(\theta,\theta',\Ufrak)}
$$
and that
\begin{equation}
\label{eq:proof1}
\dfrac{Z(\theta)}{Z(\theta')}=\ex\left\{\Psi_n(\theta,\theta',\Ufrak)\right\}\slash \ex\left\{\Phi_n(\theta,\theta',\Ufrak)\right\}\,.
\end{equation}
Now for any $(\theta,\theta')\in\Theta^2$, the expectation of the absolute value between the exact and approximate acceptance ratio is
\begin{multline}
\label{eq:proof_lem1_g}
\ex\left|\bar{a}(\theta,\theta',\Ufrak)-a(\theta,\theta')\right|
=\ex\left|\dfrac{h(\theta|\theta')}{h(\theta'|\theta)}\dfrac{p(\theta')}{p(\theta)}\dfrac{q_{\theta'}(y)}{q_{\theta}(y)}
\left(\rho_n(\theta,\theta',\Ufrak)- \dfrac{Z(\theta)}{Z(\theta')}\right)\right|\\
\leq c_p^2c_h^2 K_1 \ex\left|\rho_n(\theta,\theta',\Ufrak)-\dfrac{Z(\theta)}{Z(\theta')}\right|\,.
\end{multline}
In absence of ambiguity, we let the dependence on $(\theta,\theta',\Ufrak)$ of the random variables $\rho_n$, $\Phi_n$ and $\Psi_n$ be implicit. Using \eqref{eq:proof1}, we have:
\begin{align}
\label{eq:proof_lem1_1}
\ex\left|\rho_n-\dfrac{Z(\theta)}{Z(\theta')}\right|&=\ex\left|\Psi_n\slash \Phi_n-\ex\Psi_n\slash \ex\Phi_n\right|\,,\nonumber\\
&\leq \ex\left|\Psi_n\slash \Phi_n -\ex\left(\Psi_n\slash\Phi_n\right)\right|
+\left|\ex\left(\Psi_n\slash\Phi_n\right)-\ex\Psi_n\slash \ex\Phi_n\right|\,,\nonumber\\
&\leq \sqrt{\var\{\Psi_n\slash \Phi_n\}}+\big|\cov(\Psi_n,1\slash \Phi_n)+\ex\Psi_n\ex\left(1/\Phi_n\right)-\ex\Psi_n\slash\ex\Phi_n\big|\,,\nonumber\\
&\leq \sqrt{\var\rho_n}+\sqrt{\var\Psi_n\var\left(1\slash \Phi_n\right)}+\ex\left\{\Psi_n\left|\ex\left(1/\Phi_n\right)-1\slash\ex\Phi_n\right|\right\}\,,\nonumber\\
&\leq \sqrt{\var\rho_n}+\sqrt{\var\Psi_n\var\left(1\slash \Phi_n\right)}+\ex\Psi_n\left\{\ex\left(1/\Phi_n\right)-1\slash\ex\left(\Phi_n\right)\right\}\,.
\end{align}
Our objective is now to bound uniformly in $(\theta,\theta')$ the RHS of Eq. \eqref{eq:proof_lem1_1}. Using Eq. \eqref{eq:bounds_K2}, we have that
\begin{equation}
\label{eq_conv1}
\ex\Psi_n\left\{\ex\left(1/\Phi_n\right)-1\slash\ex\left(\Phi_n\right)\right\}\leq K_2^{C+d-1}(\eps)\left\{K_2^d(\eps)-\frac{1}{K_2^d(\eps)}\right\}\,.
\end{equation}
Defining $\Psi_{1,n}=(1/n)\sum_{k=1}^n\exp(\theta-\tdot_1)\T s(X_k^1)$ and $\Psi_{2,n}=\Psi_n/\Psi_{1,n}$, note that
\begin{multline}
\label{eq_con2}
\var\Psi_n=\var\left(\Psi_{1,n} \Psi_{2,n}\right)=\var\Psi_{1,n}\var\Psi_{2,n}+(\ex\Psi_{2,n})^2\var\Psi_{1,n}+(\ex\Psi_{1,n})^2\var\Psi_{2,n}\,,\nonumber\\
=\var\Psi_{2,n}\ex\Psi_{1,n}^2+\var\Psi_{1,n}(\ex\Psi_{2,n})^2\,.
\end{multline}
Applying Lemma \ref{lem:iid} to $\Psi_{2,n}$, leads to
\begin{equation*}
\var\Psi_{2,n}\leq K_2^{2(C-1)}(\eps)\left\{\left(1+\frac{1}{n}\right)^{C-1}-1\right\}\,,
\end{equation*}
which combined to
\begin{itemize}
\item $\ex\Psi_{1,n}^2\leq K_2^{2d}(\eps)$
\item $\ex\Psi_{2,n}\leq K_2^{C-1}(\eps)$
\item $\var\Psi_{1,n}\leq K_2^{2d}(\eps)/n$
\end{itemize}
yields
\begin{equation}
\label{eq_conf3}
\var\Psi_n\leq K_2^{2(C+d-1)}(\eps)\left\{\left(1+\frac{1}{n}\right)^{C-1}-1+\frac{1}{n}\right\}\,.
\end{equation}
Finally, combining Eq. \ref{eq_conf3} and the fact that $\var(1/\Phi_n)\leq K_2^{2d}(\eps)$, we obtain the following bound:
\begin{equation}
\label{eq_conf4}
\sqrt{\var\Psi_n\var(1/\Phi_n)}\leq K_2^{C+2d-1}(\eps)\sqrt{\left(1+\frac{1}{n}\right)^{C-1}-1+\frac{1}{n}}\,.
\end{equation}
Bounding $\var\rho_n$ follows the same technique. Because $\Phi_n$ and $\Psi_n$ are not independent, we need to rewrite $\rho_n$ in preparation for applying Lemma \ref{lem:iid} as $\rho_n=A_n B_n C_n$ where
\begin{multline*}
A_n=\Psi_{1,n},\qquad B_n={\Psi_{2,n}}\bigg\slash{\frac{1}{n}\sum_{k=1}^n\exp(\tdot_{C-1}-\tdot_C)\T s(X_k^C)},\\
C_n=\frac{\sum_{k=1}^n\exp(\tdot_{C-1}-\tdot_C)\T s(X_k^C)}{\sum_{k=1}^n\exp(\theta'-\tdot_C)\T s(X_k^C)}\,.
\end{multline*}
First note that
\begin{multline*}
\var\rho_n=\var A_n B_n\var C_n+\var A_nB_n(\ex C_n)^2+\var C_n (\ex A_nB_n)^2\\
=\var A_nB_n\ex(C_n^2)+\left(\ex A_n B_n\right)^2\var C_n\,.
\end{multline*}
Moreover, we have
\begin{multline}
\label{eq_conf5}
\var C_n=\ex\left\{\frac{\sum_{k=1}^n\exp(\tdot_{C-1}-\tdot_C)\T s(X_k^C)}{\sum_{k=1}^n\exp(\theta'-\tdot_C)\T s(X_k^C)}\right\}^2-\left\{\ex\frac{\sum_{k=1}^n\exp(\tdot_{C-1}-\tdot_C)\T s(X_k^C)}{\sum_{k=1}^n\exp(\theta'-\tdot_C)\T s(X_k^C)}\right\}^2\,,\\
\leq \frac{K_2^{2d}(\eps)}{n^2}\ex\left\{{\sum_{k=1}^n\exp(\tdot_{C-1}-\tdot_C)\T s(X_k^C)}\right\}^2-\frac{1}{K_2^{2d}(\eps)}\left\{\ex\exp(\tdot_{C-1}-\tdot_C)\T s(X^C)\right\}^2\,,\\
\leq \frac{K_2^{2d}(\eps)}{n}\var\exp(\tdot_{C-1}-\tdot_C)\T s(X^C)
+\left\{\ex\exp(\tdot_{C-1}-\tdot_C)\T s(X_k^C)\right\}^2\left\{K_2^{2d}(\eps)-\frac{1}{K_2^{2d}(\eps)}\right\}\,,\\
\leq K_2^{2(d+1)}(\eps)\left\{1+\frac{1}{n}-\frac{1}{K_2^{4d}(\eps)}\right\}
\end{multline}
and using a similarly technique, we obtain
\begin{equation}
\label{eq_conf6}
\var A_nB_n\leq K_2^{2(C+d-2)}(\eps)\left\{\left(1+\frac{1}{n}\right)^{C-1}-1\right\}\,.
\end{equation}
Combining Eqs. \eqref{eq_conf5} and \eqref{eq_conf6} with $\ex(C_n^2)\leq K_2^{2(d+1)}(\eps)$ and $(\ex A_nB_n)^2\leq K_2^{2(d+C-2)}(\eps)$, we obtain
\begin{equation}
\label{eq_conf7}
\sqrt{\var\rho_n}\leq K_2^{C+2d-1}(\eps)\sqrt{\left(1+\frac{1}{n}\right)^{C-1}+\frac{1}{n}-\frac{1}{K_2^{4d}(\eps)}}\,.
\end{equation}
Using the bounds derived in Eqs. \eqref{eq_conv1}, \eqref{eq_conf4} and \eqref{eq_conf7}, Eq. \eqref{eq:proof_lem1_1} can be written as
\begin{multline*}
\ex\left|\rho_n-\frac{Z(\theta)}{Z(\theta')}\right|\leq K_2^{C+2d-1}(\eps)\bigg\{\sqrt{\left(1+\frac{1}{n}\right)^{C-1}+\frac{1}{n}-\frac{1}{K_2^{4d}(\eps)}}\\ +\sqrt{\left(1+\frac{1}{n}\right)^{C-1}-1+\frac{1}{n}}+1-\frac{1}{K_2^{2d}(\eps)}\bigg\}\,,\\
\leq 2K_2^{C+2d-1}(\eps)\bigg\{\underbrace{\sqrt{\left(1+\frac{1}{n}\right)^{C-1}+\frac{1}{n}-1}}_{:=u_n}+\underbrace{\sqrt{1-\frac{1}{K_2^{4d}(\eps)}}}_{:=v(\eps)}\bigg\}\,,
\end{multline*}
where we have used the fact that for two positive numbers $(\alpha,\beta)$, and $\gamma>1$,
$$
\sqrt{\alpha+\beta}\leq \sqrt{\alpha}+\sqrt{\beta}\qquad\text{and}\qquad 1-\frac{1}{\gamma}\leq \sqrt{1-\frac{1}{\gamma^2}}\,.
$$
The proof is completed by noting that $u_n:=\sqrt{{C}/{n}}+o(n^{-1/2})$ and $v(\eps)=2\sqrt{d\phi_1\eps}+o(\eps^{1/2})$.
\end{proof}

\begin{proof}[Proof of Proposition \ref{prop3}]
For $n$ large enough, the delta method shows that the asymptotic distribution of $\{1/\Phi_n\}$ is
\begin{equation}
\label{appeq:2}
\frac{1}{\Phi_n}\Rightarrow g:=\norm\left(\frac{1}{\ex(\Phi_1)},\frac{\var(\Phi_1)}{n\ex(\Phi_1)^4}\right)\,.
\end{equation}
The nice benefit of this observation is that we know that denoting $\{g_n\}_n$ the sequence of distributions of $\{1/\Phi_n\}$, we have that
\begin{equation}
\label{appeq:3}
\lim_{n\to\infty}\int h(x)\rmd g_n(x)=\int h(x)\rmd g(x)\,,
\end{equation}
for any bounded measurable function $h$. Defining $f_n$ as the pdf of $(\Psi_n,\Phi_n)$ and $\alpha=\ex \Psi_n/\ex \Phi_n$, the observation \eqref{appeq:2} motivates rewriting $r_n$ \eqref{appeq:1} as follows:
\begin{multline}
\label{appeq:4}
r_n(\theta,\theta')=\int\left|\frac{\psi}{\phi}-\alpha\right|f_n(\rmd \psi,\rmd\phi)=\int\left|{\psi}{\phi}-\alpha\right|f_n(\rmd \psi\,|\,\phi)g_n(\rmd\phi)\,,\\
=\underbrace{\int\left|{\psi}{\phi}-\alpha\right|f_n(\rmd \psi\,|\,\phi)(g_n(\phi)-g(\phi))\rmd\phi}_{r_{n,1}(\theta,\theta')}+
\underbrace{\int\left|{\psi}{\phi}-\alpha\right|f_n(\rmd \psi\,|\,\phi)g(\rmd\phi)}_{r_{n,2}(\theta,\theta')}\,.
\end{multline}
The pdfs $f_n$ and $g_n$ implicitly depend on $\theta$ and $\theta'$. It is clear that given \eqref{appeq:3}, for any $(\theta,\theta')\in\Theta^2$, $r_{n,1}(\theta,\theta')\to 0$, although it is not straightforward to obtain a rate of convergence, uniformly in $(\theta,\theta')$.

Interestingly, we also have $r_{n,2}(\theta,\theta')\to 0$ and more precisely defining $W\sim g$ and $\eps\sim\norm(0,1)$, we have:
\begin{equation}
r_{n,2}(\theta,\theta')=\ex\left|\Psi_n W-\alpha\right|=\ex\left|\Psi_n \left\{\mu_1+\frac{\sigma_1}{\sqrt{n}}\eps\right\}-\alpha\right|\,,
\end{equation}
where we have defined $\mu_1=1/\ex(\Phi_1)$ and $\sigma_1^2=\var(\Phi_1)/\ex(\Phi_1)^4$. This yields
\begin{multline}
r_{n,2}(\theta,\theta')=\ex\left|\mu_1\Psi_n+\frac{\sigma_1}{\sqrt{n}}\Psi_n \eps-\alpha\right|\\
\leq \frac{1}{\ex(\Phi_1)}\ex\left|\Psi_n-\ex(\Psi_1)\right|+\frac{\sigma_1}{\sqrt{n}}\ex(\Psi_n\left|\eps\right|)\\
\leq \sqrt{\frac{\var(\Psi_1)}{n\ex(\Phi_1)^2}}+\sqrt{\frac{\var(\Phi_1)}{n\ex(\Phi_1)^4}}\ex(\Psi_n\left|\eps\right|)\\
=\frac{1}{\ex(\Phi_1)\sqrt{n}}\left\{\sqrt{\var(\Psi_1)}+\sqrt{\var(\Phi_1)}\frac{\ex(\Psi_n\left|\eps\right|)}{\ex(\Phi_1)}\right\}\,.
\end{multline}

Summarizing we have the following upper bound for $r_n$:
\begin{multline}
r_n(\theta,\theta')\leq\int\left|{\psi}{\phi}-\alpha\right|f_n(\rmd \psi\,|\,\phi)(g_n(\phi)-g(\phi))\rmd\phi+\\
\frac{1}{\ex(\Phi_1)\sqrt{n}}\left\{\sqrt{\var(\Psi_1)}+\sqrt{\var(\Phi_1)}\frac{\ex(\Psi_n\left|\eps\right|)}{\ex(\Phi_1)}\right\}\,.
\end{multline}
\end{proof}

\end{document}